\documentclass[runningheads,11pt]{article}
\usepackage{amsmath,amssymb,amsfonts,mathrsfs, amsthm}
\usepackage{algorithm}
\usepackage{algorithmic}
\usepackage[vmargin={1in,1in},hmargin={1in,1in}]{geometry}
\usepackage{authblk}
\usepackage{enumerate}
\usepackage{booktabs}
\usepackage[all]{xy}

\theoremstyle{plain}
\newtheorem{thm}{Theorem}
\newtheorem{prop}[thm]{Proposition}
\newtheorem{cor}[thm]{Corollary}
\newtheorem{lem}[thm]{Lemma}
\newtheorem{nota}{Notation}

\theoremstyle{definition}
\newtheorem{defn}[thm]{Definition}

\theoremstyle{remark}
\newtheorem{rem}{Remark}

\renewcommand{\leq}{\leqslant}
\renewcommand{\le}{\leqslant}
\renewcommand{\geq}{\geqslant}
\renewcommand{\ge}{\geqslant}

\newcommand{\eqdef}{\stackrel{\text{def}}{=}}

\newcommand{\F}{\ensuremath{\mathbb{F}}}

\newcommand{\card}[1]{\left | #1 \right |}

\newcommand{\prob}{\ensuremath{\textsf{Prob}}}

\newcommand{\code}[1]{\ensuremath{\mathscr{#1}}}
\newcommand{\Csec}{\code{C}_{\text{sec}}}
\newcommand{\Cpub}{\code{C}_{\text{pub}}}
\newcommand{\Clp}{\code{C}_{\ensuremath{\boldsymbol{\lambda}^{\bot}}}}
\newcommand{\Clpp}{\code{C}_{\ensuremath{\boldsymbol{\lambda}^{\bot}}}^{\bot}}
\newcommand{\CC}{\code{C}}
\newcommand{\DC}{\code{D}}

\newcommand{\sqc}[1]{#1^2}   %square code
\newcommand{\scp}[2]{#1\cdot #2} %scalar product
\newcommand{\cwp}{\star} % component-wise product

\newcommand{\word}[1]{\ensuremath{\boldsymbol{#1}}}
\newcommand{\av}{\word{a}}
\newcommand{\bv}{\word{b}}
\newcommand{\alphav}{\word{\alpha}}
\newcommand{\betav}{\word{\beta}}
\newcommand{\lambdav}{\word{\lambda}}

\newcommand{\cv}{\word{c}}
\newcommand{\dv}{\word{d}}
\newcommand{\ev}{\word{e}}
\newcommand{\gv}{\word{g}}
\newcommand{\hv}{\word{h}}
\newcommand{\mv}{\word{m}}
\newcommand{\pv}{\word{p}}
\newcommand{\uv}{\word{u}}
\newcommand{\vv}{\word{v}}
\newcommand{\xv}{\word{x}}
\newcommand{\yv}{\word{y}}
\newcommand{\zv}{\word{z}}

\newcommand{\mat}[1]{\ensuremath{\boldsymbol{#1}}}
\newcommand{\Dm}{\mat{D}}
\newcommand{\Gm}{\mat{G}}
\newcommand{\Pm}{\mat{P}}
\newcommand{\Pim}{\mat{\Pi}}
\newcommand{\Qm}{\mat{Q}}
\newcommand{\Rm}{\mat{R}}
\newcommand{\Sm}{\mat{S}}

\newcommand{\Gms}{\mat{G_{sec}}}
\newcommand{\Gmp}{\mat{G_{pub}}}

\newcommand{\pI}[1]{\Phi_{I}\left( #1 \right)}
\newcommand{\pIJ}[1]{\Phi_{I\setminus J} \left( #1 \right)}

\newcommand{\ff}[1]{\F_{#1}}
\newcommand{\fq}{\F_{q}}

\newcommand{\GRS}[3]{\text{\bf GRS}_{#1}(#2,#3)}

\newcommand{\map}[4]{\left\{
\begin{array}{ccc}
 #1 & \rightarrow & #2 \\
 #3 & \mapsto     & #4
\end{array}
\right.}

\title{\bf Distinguisher-Based   Attacks on Public-Key Cryptosystems
  Using Reed-Solomon Codes} 

\author{A.~Couvreur\thanks{GRACE Project, INRIA Saclay \&  LIX, CNRS UMR 7161 - \'Ecole Polytechnique, 91120 Palaiseau Cedex, France.\\
{\tt alain.couvreur@lix.polytechnique.fr}},
P.~Gaborit\thanks{XLIM, CNRS UMR 7252 - Universit\'e de Limoges, 123 avenue Albert
Thomas, 87060 Limoges Cedex, France.\\
{\tt philippe.gaborit@unilim.fr}},
V. Gauthier-Uma\~na\thanks{Faculty of Natural Sciences and Mathematics, Department of Mathematics, Universidad del Rosario, Bogot\'a, Colombia.
{\tt gauthier.valerie@ursario.edu.co}},
A.~Otmani\thanks{Normandie Univ, France;
UR, LITIS, F-76821 Mont-Saint-Aignan, France.
{\tt ayoub.otmani@univ-rouen.fr}},
J.-P.~Tillich\thanks{SECRET Project - INRIA Rocquencourt,   
78153 Le Chesnay Cedex, France. 
{\tt jean-pierre.tillich@inria.fr}}}

\begin{document}

\maketitle

\begin{abstract}
Because of their interesting algebraic properties, several authors promote the use of generalized Reed-Solomon codes in cryptography. 
Niederreiter was the first to suggest an instantiation of his cryptosystem with them but Sidelnikov and Shestakov showed 
that this choice is insecure. 
Wieschebrink proposed a variant of the McEliece cryptosystem 
which consists in  concatenating a few random columns to  a generator matrix of a secretly chosen generalized Reed-Solomon code.
More recently, new schemes appeared which are the homomorphic encryption scheme
proposed by Bogdanov and Lee, and  a variation of the McEliece cryptosystem proposed  by Baldi et \textit{al.}
which hides the generalized Reed-Solomon code by means of matrices of very low rank. 

\smallskip

In this work, we show how to mount key-recovery attacks against these public-key encryption schemes. We use the concept of distinguisher 
which aims at detecting a behavior different from the one that one would 
expect from a random code. All the distinguishers we have
built  are based on the notion of component-wise product of codes.
It results in a  powerful tool that is able to recover the secret structure of codes when they are derived from generalized Reed-Solomon codes.
Lastly, we give an alternative to Sidelnikov and Shestakov attack
by building a filtration which enables to completely recover the support and the non-zero scalars defining the secret 
generalized Reed-Solomon code.
\end{abstract}

\bigskip

\noindent \textbf{Keywords.} Code-based cryptography; generalized Reed-Solomon codes; key-recovery; distinguisher; homomorphic encryption.

\medskip

\noindent \textbf{Mathematics Subject Classication (2010):} 11T71, 94B40

\section*{Introduction}

The first cryptographic scheme using generalized Reed-Solomon codes was proposed in 1986 by 
Niederreiter  \cite{Nie86} but it was shown to be insecure in 
\cite{SidelShesta92}. The attack recovers the underlying Reed-Solomon code
allowing the decoding of any encrypted data. 
However during the past years there were several attempts
to repair this scheme.
In the present article, we focus on three modified
McEliece schemes using generalized Reed Solomon codes.
The first one
was proposed by Wieschebrink \cite{Wie06} and consists in 
choosing  a generator matrix of a generalized  Reed-Solomon code and adding to it a few random columns. It was advocated 
 that this modification avoids the Sidelnikov-Shestakov attack \cite{SidelShesta92}. 
 More recently, some of the nice algebraic properties of the Reed-Solomon codes were also used to devise 
the first public-key homomorphic encryption scheme \cite{BL12} based on coding theory.
The 
third one is another variant of the McEliece cryptosystem \cite{McEliece78}
proposed  in \cite{BBCRS11a} which uses 
this time a generator matrix of a generalized Reed-Solomon but hides its structure differently than in the McEliece cryptosystem: 
instead of multiplying by a permutation matrix, the generator matrix 
is multiplied by a matrix whose inverse is 
 of the form $ \Pim + \Rm$ 
where $\Pim$ is
a sparse matrix with row density $m\geq 1$ and $\Rm$ is a matrix of rank $z \ge 1$. 
The key point of this modification is that the public code obtained 
with this method is not anymore a generalized Reed-Solomon code and
this seems to thwart  the Sidelnikov and Shestakov attack completely. 
In the present article, we propose polynomial time attacks of these
three schemes. Notice that for Baldi \textit{et al.}'s scheme \cite{BBCRS11a},
our attack only considers the case when the matrix $\Pim$ is a permutation matrix \textit{i.e.} the case $m=1$, and
$\Rm$ is of rank $z =1$. 
We focus on these specific cases because all the parameters proposed 
 in \cite{BBCRS11a} were of this form. A good reason for these choices ($m = 1$ and $z = 1$) stems from the fact 
that the resulting schemes have the smallest public key sizes and the smallest deciphering complexity
among this class of encryption schemes.

\medskip

Contrarily to the Niederreiter's proposal \cite{Nie86} based on generalized Reed-Solomon codes, the original  McEliece cryptosystem
\cite{McEliece78}  which uses Goppa codes, has withstood many key-recovery attacks and after more than thirty years now, it still belongs 
to the very few unbroken public-key cryptosystems. 
No significant  breakthrough has been observed with respect to the problem of 
recovering the private key. For instance, the weak keys found in  \cite{Gib91,LS01} can be easily avoided. This fact has led to claim that the generator matrix of a binary Goppa 
code does not disclose any visible structure that an attacker could exploit. 
This is strengthened by the fact that Goppa codes
share many characteristics with random codes.
However, in \cite{FGOPT11a,FGOPT13}, an algorithm  that manages to 
distinguish between a random code and a Goppa code has been introduced.
This work, without undermining the security of \cite{McEliece78},
prompts to wonder whether it would be possible to devise an attack based 
on such a distinguisher. 
It turns out \cite{MP12a} that the distinguisher in \cite{FGOPT11a,FGOPT13} has an equivalent but simpler description
in terms of the component-wise product of codes. This notion was first put forward in coding theory to unify many 
different algebraic decoding algorithms \cite{Pel92,Kot92a}. Recently, it was used in \cite{CMP11,CMP11:DCC}
to study the security of cryptosystems based on Algebraic-Geometric codes.
Component-wise powers of codes are also studied in the context of secure multi-party computation (see for example \cite{Cramer09,Cramer11}).
This distinguisher is even more powerful in the case of Reed-Solomon codes 
than for Goppa codes.
Indeed, whereas for Goppa codes it is only successful for rates close to $1$, it
can distinguish Reed-Solomon codes of any rate from random codes. 

\medskip

In the specific case of \cite{BL12}, the underlying public code is a modified Reed-Solomon code obtained
from  the insertion of a zero submatrix in the Vandermonde generating matrix  defining it and in this case, the aforementioned distinguisher leads to an attack that is different from the one found independently by Brakerski in \cite{brakerski13}. More exactly,
we present a key-recovery attack on the Bogdanov-Lee homomorphic scheme based on the version of our distinguisher presented in \cite{MP12a}. Our attack runs in polynomial time and is efficient: it only amounts to calculate the ranks of certain matrices derived from the public key.
In \cite{BL12} the columns that define 
 the zero submatrix are kept secret and form a set $L$. We give here a distinguisher that detects if one or several columns belong
to $L$ or not. It is constructed by considering the code  generated by component-wise products of codewords of the public code
(the so-called ``square code"). This operation is applied to punctured versions of this square code obtained by picking a subset 
 $I$ of the whole set indexing the columns. It turns out that the dimension of
 the punctured square code  is directly related to  the cardinality of
 the intersection of $I$ with $L$. 
This gives a way to recover the full set $L$ 
allowing the decryption of  any ciphertext. 

\medskip

We also propose another cryptanalysis against 
the variant of the McEliece cryptosystem \cite{McEliece78} proposed  in \cite{BBCRS11a}.
As explained above, the public code obtained 
with this method is not anymore a generalized Reed-Solomon code (GRS for short).
On the other hand, it contains a very large secret GRS code. 
We present an attack that is based on a distinguisher which is able to identify elements of this secret code. This distinguisher is again 
derived from considerations about the dimension of component-wise products of codes. Once this secret code is obtained, it is then possible to completely recover the initial GRS code by using the 
square-code construction as in \cite{Wie10}. We are then able to decode any ciphertext.

\medskip

Finally, we also cryptanalyze  the first variant of the McEliece cryptosystem based on GRS codes \cite{Wie06}. 
We show here how a refinement of our distinguisher 
permits to recover the random columns added to the generator matrix of the GRS code. Once these column positions are recovered, the Sidelnikov and Shestakov attack can be used on the non-random
part of the generator matrix to completely break the scheme.
It should also be pointed out that the properties of Reed-Solomon codes with respect to the (component-wise) product of codes have 
already been used to cryptanalyze a McEliece-like scheme  \cite{BL05} based on subcodes of Reed-Solomon codes
\cite{Wie10}. The use of this product is nevertheless different in
\cite{Wie10} from the way we use it here.
Note also that our attack is not an 
adaptation of  the Sidelnikov and Shestakov approach \cite{SidelShesta92}. 
Our approach is completely new: it illustrates how a distinguisher
that detects an abnormal behavior can be used to recover a private
key. 

\medskip

To demonstrate further the power of our approach, we give an alternative to Sidelnikov and Shestakov's 
way \cite{SidelShesta92} to fully recover the structure of a generalized
Reed-Solomon codes. Our new attack uses the code product to build a decreasing
chain of subcodes resulting to a code of very small dimension which shares the same support
as the original secret generalized Reed-Solomon code and for which the
structure is very simple to recover. This achievement is obtained by
repeatedly solving linear systems. 
The resulting complexity is $O(k^2 n^3+k^3n^2)$ operations in
the underlying field. This attack is more complex than the 
original Sidelnikov and Shestakov but, because it does not rely on the computation of
minimum codewords as in \cite{SidelShesta92}, 
it might be applied to other families of codes such as Reed-Muller codes.  
This is in particular the case for wild Goppa codes \cite{BLP10:wild} as shown in 
the paper \cite{COT:EC14} where this technique was further developed  and applied to wild Goppa codes 
defined over quadratic extensions. It gave for the first time a polynomial time attack 
on a McEliece cryptosystem based on non-binary Goppa codes. This recent result highlights the potential power 
of this method in cryptography.

\paragraph{\bf Organization of the paper.} In Section~\ref{sec:basics} we recall relevant notions from coding theory. 
In Section~\ref{sec:Wiesch}, we show that adjunction of random columns to a generalized Reed-Solomon codes as advocated in \cite{Wie06} does not improve the security of McEliece-like cryptosystems based on Reed-Solomon codes. 
In Section~\ref{Sec:BL} we describe 
the cryptanalysis of the homomorphic cryptosystem introduced by Bogdanov and Lee in \cite{BL12}. 
Section~\ref{sec:schemeit} describes the cryptosystem proposed in \cite{BBCRS11a} and explains the reasons 
why this scheme is insecure. In Section~\ref{sec:GRS_attack} we give another way to attack a scheme based on generalized Reed-Solomon codes, and lastly we conclude the paper.

\section{Reed-Solomon Codes and the Square Code Construction}
  \label{sec:basics}

We recall in this section a few relevant results and definitions from coding theory and bring in the
fundamental notion which is used in both attacks, namely the square code construction.
Generalized Reed-Solomon codes (GRS in short)  form a special case of codes with a very powerful low complexity decoding algorithm.
It will be convenient to use the definition of these codes as
{\em evaluation codes}.

\begin{defn}[Generalized Reed-Solomon code] \label{defGRS}
Let $k$ and $n$ be integers such that $1 \leqslant k < n \leqslant q$ where $q$ is a power of
a prime number.
The generalized Reed-Solomon code $\GRS{k}{\xv}{\yv}$ of dimension $k$ is associated to a pair
$(\xv,\yv) \in \fq^n \times \fq^n$ where $\xv$ is an $n$-tuple of distinct elements of
$\fq$  and $\yv$
an $n$--tuple of arbitrary nonzero elements in $\fq$. 
The code
$\GRS{k}{\xv}{\yv}$ is 
defined as: 
$$
\GRS{k}{\xv}{\yv} \eqdef
\Big\{(y_1p(x_1),\dots{},y_np(x_n)) : p \in \fq[X], \deg p < k\Big\}.
$$
\end{defn}

\begin{rem}
 Reed-Solomon codes correspond to the case
where $y_i=1$ for all $i$.
\end{rem}

The first work that suggested to use GRS codes in a public-key cryptosystem scheme 
was \cite{Nie86}. But Sidelnikov and Shestakov  
discovered in \cite{SidelShesta92} 
that this scheme is insecure. They namely showed that for any GRS code it is 
possible  to recover in polynomial time a  couple $(\xv,\yv)$ which defines it.
This is all that is needed to decode efficiently such codes and is therefore enough to break the Niederreiter cryptosystem suggested in 
\cite{Nie86} or any McEliece type cryptosystem \cite{McEliece78} that uses GRS codes instead of binary 
Goppa codes.

\begin{defn}[Componentwise products]
  Given two vectors $\av=(a_1, \dots, a_n)$ and $\bv=(b_1, \dots, b_n)
  \in \fq^n$, we denote by $\av \cwp \bv $ the componentwise product
 $$
 \av \cwp \bv \eqdef (a_1b_1,\dots{},a_n b_n)
 $$
\end{defn}

The star product should be distinguished from a more useful operation
in coding theory, namely the canonical inner product:
\begin{nota}
  Given $\av, \bv\in \fq^n$, the inner product $\scp{\av}{\bv}$ is defined as
  $$
  \scp{\av}{\bv}\eqdef \sum_{i=1}^n a_i b_i.
  $$
\end{nota}

\begin{defn}[Product of codes \& square code]
Let $\code{A}$ and $\code{B}$ be two codes of length $n$. The
\emph{star product code} denoted by $\code{A} \cwp \code{B}$ of $\code{A}$ and $\code{B}$ is \emph{the vector space
spanned by all products} $\av \cwp \bv$ where $\av$ and $\bv$ range over $\code{A}$ and $\code{B}$ respectively.
When $\code{B} = \code{A}$ then  $\code{A} \cwp \code{A}$ is called the \emph{square code} of $\code{A}$
and is rather denoted by $\sqc{\code{A}}$.
\end{defn}

It is clear that $\code{A} \cwp \code{B}$ is also generated by the $\av_i \cwp \bv_j$'s where the $\av_i$'s and the
$\bv_j$'s form a basis of $\code{A}$ and $\code{B}$ respectively.
Therefore, we have the following result.
\begin{prop}\label{prop:dimprod}
Let $\code{A}$ and $\code{B}$ be two codes of length $n$, then
\begin{enumerate}
\item\label{item:prod}  $\dim(\code{A} \cwp \code{B}) \le \dim(\code{A})
  \dim(\code{B})$;
\item\label{item:square} $\displaystyle \dim (\sqc{\code{A}}) \leq \binom{\dim(\code{A})+1}{2}$.
\end{enumerate}
\end{prop}

\begin{prop}\label{prop:complexity}
  Let $\code{A} \subset \fq^n$ be a code of dimension $k$.
  The complexity of the computation of a basis of
  $\sqc{\code{A}}$ is $O(k^2n^2)$ operations in $\fq$.
\end{prop}

\begin{proof}
  The computation, consists first in the computation of ${k+1 \choose 2}$
  generators of $\sqc{\code{A}}$. This computation costs $O(k^2n)$ operations.
  Then, we have to apply a Gaussian elimination to
  a ${k+1 \choose 2}\times n$ matrix, which costs $O(k^2n^2)$ operations.
  This second step is dominant, which yields the result.
\end{proof}

The importance of the square code construction will become clear when we compare the dimensions 
of square codes obtained through a \emph{structured} code and random code and one major question is to know 
what one should expect. The following Proposition~\ref{prop:square} shows that when applied to GRS codes, 
the dimension of the square code is roughly  twice as large as the dimension of the underlying code.
This fact has been already observed in \cite{Wie10} in a cryptanalytic setting. A proof can also be found in \cite[Proposition 10]{MMP11a:DCC}.

\begin{prop}\label{prop:square} For $k\leq (n+1)/2$, we have
$
\GRS{k}{\xv}{\yv}^2 = \GRS{2k-1}{\xv}{\yv\cwp\yv}.
$
\end{prop}

\begin{proof}
 This follows immediately from the definition of a GRS code as an evaluation code since
 the star product of two elements $\cv=(y_1 p(x_1),\dots,y_np(x_n))$ and $\cv'=(y_1 q(x_1),\dots,y_nq(x_n))$ of $\GRS{k}{\xv}{\yv}$ where
 $p$ and $q$ are two polynomials of degree at most $k-1$ is of the form 
 $$\cv \cwp \cv' 
 = \big(y_1^2 p(x_1)q(x_1),\dots,y_n^2 p(x_n)q(x_n)\big)
 =\left(y_1^2r(x_1),\dots,y_n^2r(x_n)\right)
 $$
 where $r$ is a polynomial of degree $\leq 2k-2$. Conversely, any element of the form $\left(y_1^2r(x_1),\dots,y_n^2r(x_n)\right)$
 where $r$ is a polynomial of degree less than or equal to $2k-2$  is a linear combination of star products of two elements of $\GRS{k}{\xv}{\yv}$.
\end{proof}

This proposition shows
that the square code is only of dimension $2k-1$ when $2k-1 \leq n$.
This property can also be used in the case $2k-1 >n$. To see this,
consider the dual of the Reed-Solomon code, which is itself a Reed-Solomon code 
\cite[Theorem 4, p.304]{MacSloBook}

\begin{prop}\label{pr:dual} 
$
 \GRS{k}{\xv}{\yv}^\perp = \GRS{n-k}{\xv}{\yv'} 
$
 where the length of $\GRS{k}{\xv}{\yv}$ is $n$ and $\yv'$ is a certain element of $\fq^n$ depending only on $\xv$ and $\yv$.
\end{prop}

This result is clearly different from what would be obtained if random linear
codes were taken. Indeed,
we expect that the square code when applied to a random linear code of dimension $k$ should be a code of dimension of
order $\min\left\{\binom{k+1}{2},n\right\}$. Actually it can be shown by the proof technique of 
\cite{FGOPT11a,FGOPT13}  the following result (see also \cite{MP12a}).

\begin{prop}[\cite{FGOPT11a,FGOPT13}] \label{prop:square_random}
Let $k$ and $n$ be non-negative integers such that $k = O(n^{1/2})$ and consider a random $(n-k) \times (n-k)$ 
matrix $\Rm$ where each entry is independently and uniformly drawn from $\fq$. 
Let $\code{R}$ be the linear code defined by the generator matrix $\left ( \Im_k \mid \Rm \right)$ where $\Im_k$ is the $k \times k$ identity matrix. 

\smallskip

For any $\varepsilon$ such that $0 < \varepsilon < 1$ and any $\alpha > 0$, we have as $k$ tends to $+\infty$:
$$
\prob\left(\dim\big(\sqc{\code{R}}\big) \leq \binom{k+1}{2}\Big(1 -\alpha k^{-\varepsilon} \Big) \right) = \circ (1)
$$

\end{prop}

Therefore $\GRS{k}{\xv}{\yv}$ can be distinguished from 
 a random linear code of the same dimension by computing the dimension of the associated square codes.
This phenomenon was already observed in \cite{FGOPT11a,FGOPT13} for $q$-ary alternant codes (in particular Goppa codes) at very high rates. 
Let us note that even when  $2k-1 > n$  it is still possible to distinguish GRS codes from random codes by focusing on
 $\left(\GRS{k}{\xv}{\yv}^\perp\right)^2$. We have in this case:
 $$
\left(\GRS{k}{\xv}{\yv}^\perp\right)^2
 =
\GRS{n-k}{\xv}{\yv'}^2
 = 
\GRS{2n-2k-1}{\xv}{\yv'\cwp\yv'}\sqc{\code{C}}
 $$
which is a code of dimension $2n-2k-1$.

\medskip

The star product of codes  has been used for the first time by Wieschebrink to 
cryptanalyze a McEliece-like scheme \cite{BL05} based on subcodes of Reed-Solomon codes
 \cite{Wie10}. The use of the star product is nevertheless different in \cite{Wie10} from the way we use it here. In Wieschebrink's paper,
 the star product is used to identify  for a certain subcode  $\code{C}$ of a GRS code $\GRS{k}{\xv}{\yv}$
 a possible pair $(\xv,\yv)$. This is achieved by computing $\sqc{\code{C}}$ which turns out to be $\GRS{k}{\xv}{\yv}^2 = \GRS{2k-1}{\xv}{\yv \cwp \yv}$. The Sidelnikov and Shestakov algorithm is then 
 used on  $\sqc{\code{C}}$ to recover a possible $(\xv,\yv\cwp\yv)$ pair to describe $\sqc{\code{C}}$ as a GRS 
 code, and hence, a pair $(\xv,\yv)$  is deduced for which $\code{C}
 \subset \GRS{k}{\xv}{\yv}$.

\section{Wieschebrink's Encryption Scheme}
   \label{sec:Wiesch}

In \cite{Wie06}  Wieschebrink suggests a variant of the McEliece cryptosystem based
on GRS codes whose purpose was to resist to the Sidelnikov--Shestakov
attack.
The idea of this proposal is to use the generator matrix of a
GRS code over $\fq$ 
in which a small number of randomly chosen columns  are inserted.
More precisely, let $\Gm$ be a generator matrix of a GRS
    code of length $n$ and dimension $k$ 
    defined over $\fq$. Let $C_1, \ldots , C_r$ be $r$ column vectors
    in $\fq^k$ drawn uniformly at
    random and let $\Gm'$ be the matrix obtained by concatenating
    $\Gm$ and the columns $C_1,\ldots , C_r$. 
    Choose   $\Sm$ to be a $k\times k$ random invertible matrix
and let $\Qm$ be a an $(n+r)\times (n+r)$ permutation matrix. 
The public key of the scheme is
$$\Gmp \eqdef
\Sm^{-1} \Gm ' \Qm ^{-1}.$$
This cryptosystem can be cryptanalyzed if a description of the GRS code
can be recovered from $\Gmp$.
We give here a way to break this scheme in polynomial time which relies on  two  ingredients.
The first one is given by
\begin{lem}
  \label{lem:dimC'^2}
     Let $\Gm'$ be a $k\times (n+r)$--matrix obtained by inserting $r$
     random columns in a generator matrix of an $[n, k]$ GRS code $\CC$.
     Let $\CC'$ be the corresponding code. Assume that $k<(n-r)/2$, then
     $$
       2k-1 ~ \leqslant ~ \dim \sqc{\CC'} ~ \leqslant ~ 2k-1+r.
     $$
\end{lem}

\begin{proof}
The first inequality comes from the fact that puncturing $\CC'^2$ at the $r$
  positions corresponding to the added random columns yields the code $\CC^2$
  which is the square of an $[n, k]$ GRS code and hence an $[n, 2k-1]$ GRS code.
To prove the upper bound, let 
$\DC$ be the code with generator matrix $\Dm$ obtained from $\Gm'$ by replacing the $C_{i}$'s columns  by all-zero columns and 
let 
$\DC'$ be the code with generator matrix $\Dm'$ obtained by replacing in $\Gm'$ all columns which are not the $C_{i}$'s by zero columns.
Since $\Gm' = \Dm+\Dm'$ we have
\begin{equation}\label{eq:sum}
\CC' \subset \DC + \DC'.
\end{equation}
Therefore 
\begin{eqnarray*}
\sqc{\CC'} & \subset & \sqc{\left(\DC + \DC'\right)}\\
& \subset & \sqc{\DC} + \sqc{\DC'} +  \DC\cwp \DC' \\
& \subset & \sqc{\DC} + \sqc{\DC'}
\end{eqnarray*}
where the last inclusion comes from the fact that $ \DC\cwp \DC'$ is the zero subspace
since $\DC$ and $\DC'$ have disjoint supports.
The right-hand side inequality follows immediately from this, since 
$\dim \sqc{\DC} = 2k-1$ and $ \dim \sqc{\DC'} \leq r$. 
\end{proof}

\begin{rem}
Actually the right-hand inequality of 
Lemma~\ref{lem:dimC'^2} is sharp and
we have observed experimentally that if $2k-r-1 < n$ then,
we almost always get
\begin{equation} \label{wie:rank}
\dim \sqc{\CC'} = 2k-1+r.
\end{equation}
For instance with values for $q$, $n$ and $r$ like those proposed in \cite{Wie06} and choosing $k = (n-r)/2 -1$ 
we observed with $1000$ random instances that Equation~\eqref{wie:rank} was always satisfied.
\end{rem}

This will be useful to detect the positions which correspond to the $C_{i}$'s.
We call such positions the {\em random positions} whereas the other positions 
are referred to as the {\em GRS positions}. We use in this case a
shortening trick which relies upon the following well-known lemma.

\begin{lem}[\cite{HPbook}]
 \label{fa:shortening}
Shortening a GRS code of parameters $[n,k]$ in $\ell \leqslant k$ positions
 gives
 a GRS code with parameters $[n-\ell,k-\ell]$.
\end{lem}
 
An attack easily follows from these facts. 
First of all, let us consider the case when
$2k-1+r \leqslant n$, then consider $\CC_i'$ which is
the punctured $\CC'$ code at position $i$.
Two cases can occur:
\begin{itemize}
\item {\em $i$ belongs to the random positions}, then we expect that the dimension of $\sqc{\CC_i'}$ is given by
$$
      \dim \sqc{\CC_i'} = 2k-2 + r.
    $$
    since $\CC_i'$ is a GRS code of dimension $k$ with $r-1$ random columns inserted in its generator matrix hence 
    $\dim \sqc{\CC_i'} = \dim \sqc{\CC'} + r-1 = 2k-2+r$ with a high probability.
\item {\em $i$ belongs to the GRS positions}, then $\CC_i'$
    is a GRS code of dimension $k$ with $r$ random columns inserted in its  generator matrix so that
\[
\dim \sqc{\CC_i'} = 2k-1 + r.
\] 
\end{itemize}
This gives a straightforward way to distinguish between the random positions and the GRS positions.

Consider now the case where $2k-1+r > n$. The point is to shorten $\CC'$ in $a$ positions,
then, thanks to Lemma~\ref{fa:shortening}, the same principle can be applied. 
Here  $a$ is chosen such that $a<k$ and
   $2(k-a)-1+r < n-a$ so that $a > 2k-1+r-n$.
Notice that these conditions on $a$ can be met as soon as
$
k> 2k +r -n$ that is to say $n > k+r
$, which always holds true.
Among these $a$ positions, $a_0$ of them are random positions and $a_1\eqdef a-a_0$ are
GRS positions. This yields an $a_0$--codimensional subcode of a GRS code of
parameters $[n-a_1, k-a_1]$ to which 
$r-a_{0}$ random positions have been added (or more precisely this yields 
a code with generator matrix given by the generator matrix of an
 $a_0$--codimensional subcode of a GRS code of
size $(k-a_{1})\times (n-a_{1})$ with $r-a_{0}$ random columns added to it).
Let $I_a$ be a set of $a$ positions
and denote by $\CC'_{I_a}$ the code $\CC'$ shortened in these positions.
Using the previous results, we get that with  high probability,
$$
   \dim \sqc{{\CC'_{I_a}}} = 2(k-a_1)-1 +r-a_0
$$
By this manner we get the value of $2a_1 +a_0$ and since $a=a_1+a_0$ is already
known we can deduce the values of $a_0$ and $a_1$.
To identify which positions of $\CC'_{I_a}$ are random positions and which ones
are GRS positions we just 
use the previous approach  by shortening $\CC'_{I_a}$ in an additional position
and checking whether or not the dimension decreases by one or two.
This approach has been implemented in Magma and leads to identify easily all
the random columns
for the parameters suggested in \cite{Wie06}. 
After identifying the random
columns in the public generator matrix, it just 
remains to puncture the public code at these positions and to apply the
Sidelnikov-Shestakov attack to completely 
break the scheme proposed in \cite{Wie06}. 
The complexity of guessing the random columns in the public generator matrix is hence given by the complexity of computing
the rank of $n+r$ matrices of size $\binom{k+1}{2} \times (n+r-1)$, that is to say 
$O\left( (n+r) k^2 (n+r)^2\right) = O\left( k^2 (n+r)^3\right)$ operations in the field $\fq$.

If, moreover, we assume that $2k-1+r > n$
as it is the case in \cite{Wie06} then in a worst-case scenario we would guess only one position among the random ones so that we have to iterate at most $r$ times the previous procedure. 
The complexity of the Sidelnikov-Shestakov attack \cite{SidelShesta92} is $O\left(k^3+k^2 n\right)$ and is negligible compared to the other calculations.
Thus, the complexity of the attack is $O\left( k^2 r(n+r)^3\right)$ operations in the field $\fq$.
In Table~\ref{tab:wie_time} we gathered the running times of the attack implemented in  \textsc{Magma}~(V2.19-9) \cite{MAGMA} and obtained with an 
Intel\textsuperscript{\textregistered} Xeon 2.90GHz.

\begin{table}
\begin{center}
\begin{tabular}{@{}*{5}{r}@{}}
\toprule 
$q$ & $n$ &$k$ & $r$ & Time (in seconds) \\
 \midrule
$128$ & $128$ & $79$ & $20$   & $9.22$ \\
$256$ & $256$ & $169$ & $39$ & $103.84$ \\
$512$ & $384$ & $245$ & $64$  & $517.78$ \\
$512$ & $512$ & $335$ & $83$  & $1517.98$\\
 \bottomrule
\end{tabular}
\end{center}
\caption{Average running time of the attack against Wieschebrink encryption scheme \cite{Wie06} with $N = 100$ trials.} 
\label{tab:wie_time}
\end{table}

\section{Bogdanov-Lee Homomorphic Cryptosystem} 
  \label{Sec:BL}

\subsection{Description of the scheme}\label{ss:scheme}

The cryptosystem proposed by Bogdanov and Lee in \cite{BL12} is a 
public-key homomorphic encryption scheme based on linear codes. 
It encrypts a plaintext $m$ from $\fq$  into a ciphertext $\cv$ that belongs to  
$\fq^n$ where $n$ is a given integer satisfying $n<q$.
The key generation requires two non-negative 
integers $\ell, k$ such that $3\ell < n$ and $\ell <k$ together with a
subset $L \subset \{1,\ldots,n\}$ of cardinality $3\ell$.
A set of $n$ distinct elements $x_1,\dots, x_n$ from $\fq^{\times}$ are 
generated at random.  They serve to construct a $k \times n$ matrix $\Gm$
whose $i$-th column $\Gm^T_i$ ($1\leqslant i \leqslant n$)
is defined by
$$
\Gm_i^T \eqdef \left\{
\begin{array}{ll}
(x_i, x_i^2, \dots, x_i^{\ell}, 0, \dots , 0 )                  &
\text{if } i \in L\\
& \\
(x_i, x_i^2, \dots, x_i^{\ell}, x_i^{\ell+1}, \dots , x_i^{k} ) & \text{if } i \notin L
\end{array} \right. ,$$
where the symbol $^T$ stands for the transpose.The cryptosystem is defined as follows: 
\begin{enumerate}
	\item \textbf{Secret key.} $(L, \Gm)$.
	\item \textbf{Public key.} $\Pm \eqdef \Sm \Gm $ where $\Sm$
          is a $k \times k$ random invertible  matrix over $\fq$.
	\item \textbf{Encryption.} The ciphertext $\cv \in \fq^n$
          corresponding to
          $m \in \fq$ is obtained as $\displaystyle \cv \eqdef \xv \Pm
          + m \textbf{1} +  \ev$ where $\textbf{1} \in \fq^n$ is the
          all-ones row vector, 
  $\xv$ is picked uniformly at random in $\fq^k$  and $\ev$
                          in $\fq^n$  by choosing its components according to a certain distribution $\tilde{\eta}$.
	\item \textbf{Decryption.} Solve the following linear system
          with unknowns $\yv \eqdef(y_1,\dots{},y_n) \in \fq^n$: 
	\begin{equation} \label{decryptionsystem}
          \Gm \yv^T   =  0 \text{, } \sum_{i \in L} y_i  =  1\text{ and }
          y_i=0   \textrm{ for all }  i \notin L.
                  \end{equation} 
          The plaintext is then $m = \displaystyle  \sum_{i = 1}^n y_i c_i$. 
\end{enumerate}

Let us explain here why  the  decryption algorithm  outputs  the correct plaintext when 
$\ell$ and $n$ are chosen such that  the entry $e_i$ at position $i$ of the error vector is zero when $i \in L$.
If this property on $\ev$ holds, notice that  the linear system~\eqref{decryptionsystem} has $3\ell$ unknowns and $\ell +1$ equations and since it is by construction of rank $\ell+1$, it always admits at least one solution. Then observe that
\begin{eqnarray*}
\sum_{i = 1}^n y_i c_i &=& (\xv \Pm + m \textbf{1} + \ev) \yv^T\\
& =&  (\xv \Pm + m \textbf{1} ) \yv^T \;\;\;\text{   (since $e_i=0$ if $i\in L$ and $y_i=0$ if $i \notin L$)}\\
& = &\xv \Sm \Gm \yv^T + m \sum_{i=1}^n y_i \\
& =  & m \;\;\;\text{   (since $\Gm \yv^T=0$ and  $\sum_{i=1}^n y_i=1$)}.
\end{eqnarray*}

\smallskip

The decryption algorithm will output  the correct plaintext when 
$\ell$ and $n$ are chosen such that  the entry $e_i$ at position $i$ of the error vector is zero when $i \in L$.
The distribution $\tilde{\eta}$ which is used to draw at random the coordinates of $\ev$ is chosen such that this
property holds with very large probability.  More precisely, the parameters $k$, $q$, $\ell$ and the noise distribution $\tilde{\eta}$ are chosen such that
 $q = \Omega\left(2^{n^\alpha}\right)$,
 $k=\Theta\left(n^{1-\alpha/8}\right)$, $\ell =
 \Theta\left(n^{\alpha/4}\right)$ and
 the noise distribution $\tilde{\eta}$ is the $q$-ary symmetric channel
 with noise rate\footnote{It means that $\prob(e_i=0)=1-\eta$ and $\prob(e_i=x)=\frac{\eta}{q-1}$ for any $x$ in $\fq$ different from zero.} 
 $\eta = \Theta\left(1/n^{1-\alpha/4}\right)$
 where  $\alpha \in [0,\frac{1}{4}]$.
To understand why these parameters work, we refer to \cite[\S 2.3]{BL12}.

\subsection{An efficient key-recovery attack}\label{attackBL}

We present here an attack that is different from Brakerski's one \cite{brakerski13}. 
Ours consists in first recovering the secret set $L$ and from
here, one finds directly a suitable vector $\yv$ by solving the system
\begin{equation} \label{decryptionsystempub}
\Pm \yv^T     =  0,\;\sum_{i \in L} y_i  =  1,\;
       y_i=0   \textrm{ for all }  i \notin L.
\end{equation}
Indeed, requiring that $\Pm \yv^T = 0$  is equivalent to  the
equation $\Gm \yv^T=0$ since, by definition, $\Pm = \Sm \Gm$ and since $\Sm$ is invertible.
Therefore, \eqref{decryptionsystempub} is equivalent to the
``secret'' system  \eqref{decryptionsystem}. An attacker may therefore recover $m$ without even knowing $\Gm$ just 
by outputting $\sum_i y_i c_i$ for any solution $\yv$ of \eqref{decryptionsystempub}.
In what follows, we will explain how   $L$ can be recovered from $\Pm$ in 
polynomial time. 

Our attack which recovers $L$ relies heavily on the fact that the public matrix may be viewed as a generator matrix of a code 
$\code{C}$ which 
is quite close to a generalized Reed-Solomon code (or to a Reed-Solomon code if a row consisting only of $1$'s is added
to it). Notice that any punctured version of the code has also this property (a punctured code consists in 
keeping only a fixed subset of positions in a codeword). More precisely, let us introduce 
\begin{defn}
For any $I \subset \{1,\dots{},n\}$ of cardinality $\card{I}$, the
restriction of a code $\code{A}$ of length $n$ is the subset of
$\fq^{|I|}$ defined as 
$\code{A}_I \eqdef \Big \{  \vv  \in \fq^{\card{I}} \mid \exists \av \in \code{A}, \vv = (a_i)_{i \in I} \Big \}.
$
\end{defn}

The results about the unusual dimension of the square of a Reed-Solomon codes which are given in 
Section \ref{sec:basics} prompt us to study the dimension of the square code $\sqc{\code{C}}$ or more generally the
dimension of $\code{C}_I^2$. When $I$ contains no positions in $L$, then $\code{C}_I$ is nothing but a generalized Reed-Solomon 
code and we expect for $\CC^2$ a dimension of $2k-1$ when $|I|$ is larger than $2k-1$. On the other hand, when there are positions in $I$
which also belong to $L$ we expect the dimension to become bigger and  the dimension of $\sqc{\code{C}}$ to behave 
as an increasing function of $|I \cap L|$. This is exactly what happens as shown in the proposition below.
 
\begin{prop} \label{prop:dsg}
Let $I$ be a subset of $\{1,\dots{},n\}$ and set $J
\eqdef I \cap L$. 
If the cardinality of $I$ and
$J$ satisfy $\card{J} \leqslant
\ell-1$ and $\card{I} -  \card{J} \geqslant 2k$ then
\begin{equation}
\dim(\code{C}_I ^2) = 2k - 1 + \card{J}.
\end{equation}
\end{prop}

The proof of this proposition can be found  in  Appendix \ref{sec:proof_prop:dsg}.
An attacker can exploit this proposition to mount a distinguisher that recognizes whether a given position belongs 
to the secret set $L$. At  first  a set $I$ which satisfies with high probability the assumptions of
Proposition \ref{prop:dsg} is randomly chosen. Take for instance $|I|=3k$. Then $k_I \eqdef \dim(\code{C}_I^2 )$ is computed. Next,  one element $x$ is removed from $I$ to get a new set
$I'$ and  $k_{I'} = \dim(\code{C}_{I'}^2 )$ is computed. The only
two possible cases are either $x \notin L$ then $k_{I'} = k_{I}$ or $x
\in L$ and then $k_{I'} = k_{I}-1$.
By repeating this procedure, the whole set $J=I \cap L$ is easily recovered. The
next step now is to find all the elements of $L$ that are not in
$I$. One solution is to exchange one element in $I \setminus J$ by another element in
$\{1,\dots{},n\} \setminus I$ and compare the values of $k_I$. If it
increases, it means that the new element belongs to $L$. At the
end of this procedure the set $L$ is totally recovered. This probabilistic algorithm is obviously of  polynomial time complexity
and breaks completely the homomorphic scheme suggested in \cite{BL12}.

\subsection{Inherent weakness of the scheme}

The purpose of this section is to explain why the homomorphic scheme of \cite{BL12} leads
in a natural way to define codes whose square code has an abnormal low dimension. This property which seems 
inherent to the scheme implies that there is little hope to propose a reparation. This fact was also observed in \cite{brakerski13}. 
The point of \cite{BL12} is to define a code which is homomorphic for  addition over 
$\fq$ (all linear codes do the job here) but also protohomorphic for the multiplication over 
$\fq$ \cite[Claim 3.5]{BL12}. This property holds for their scheme, because there is a solution
$\yv$ of \eqref{decryptionsystem} which satisfies for two ciphertexts $\cv$ and $\cv'$ in $\fq^n$ corresponding 
respectively to the plaintexts $m$ and $m'$ in $\fq$:
\begin{equation}\label{eq:homomorphic}
\scp{\yv}{(\cv \cwp \cv')} = mm'
\end{equation}
Recall that $\cv$ and $\cv'$ are given by
\begin{eqnarray}
\cv & = & \xv \Pm + m \textbf{1} +  \ev\\
\cv' & = & \xv' \Pm + m' \textbf{1} +  \ev'
\end{eqnarray}
where $\ev$ and $\ev'$ are error vectors whose support does not intersect $L$.
We also know that $\yv$ satisfies:
\begin{enumerate}
\item $\Gm \yv^T=0$;
\item $\sum_{i=1}^n y_i=1$;
\item $y_i=0$ if $i \notin L$
with $\Pm$ and $\Gm$ related by a multiplication of an invertible matrix $\Sm$, \textit{i.e.}
$\Pm = \Sm \Gm$.
\end{enumerate}
We deduce from this
\begin{eqnarray*}
\yv(\cv \cwp \cv')^T & = &\yv \left( (\xv \Pm + m \textbf{1} +  \ev)\cwp (\xv' \Pm + m' \textbf{1} +  \ev')\right)^T\\
& = & \yv \left( \Pm^T \xv^T \cwp  \Pm^T \xv'^T + \Pm^T \xv^T \cwp m' \textbf{1}^T +  
\Pm^T \xv'^T \cwp m \textbf{1}^T  + m\textbf{1}^T \cwp m' \textbf{1}^T \right)\\
&  & + \yv\left( \ev^T  \cwp ( \Pm^T \xv'^T + m' \textbf{1}^T + \ev'^T)\right)+  \yv \left( (\Pm^T\xv^T + m \textbf{1}^T)\cwp  \ev'^T \right)
\end{eqnarray*}
The terms $\yv\left( \ev^T  \cwp ( \Pm^T \xv'^T + m' \textbf{1}^T + \ev'^T)\right)$ and 
$\yv \left( (\Pm^T\xv^T + m \textbf{1}^T)\cwp  \ev'^T \right)$ are equal to zero because the support of $\yv$ is contained in  $L$
and $\ev^T  \cwp ( \Pm^T \xv'^T + m' \textbf{1}^T + \ev'^T)$, $(\Pm^T\xv^T + m \textbf{1}^T)\cwp  \ev'^T $ have their support
outside $L$. The terms $\yv (\Pm^T \xv^T \cwp m' \textbf{1}^T)=m' \yv \Gm^T\Sm^T\xv^T$ and 
$\yv (\Pm^T \xv'^T \cwp m \textbf{1}^T)=m \yv \Gm^T\Sm^T\xv'^T$ are equal to $0$ from Condition (i) on $\yv$ given
above. Therefore in order to ensure \eqref{eq:homomorphic} we need that
\begin{equation}
\label{eq:condition_square}
\yv \left( \Pm^T \xv^T \cwp  \Pm^T \xv'^T\right)=0.
\end{equation}
has a non zero solution whose support is contained in $L$.
Let $\code{C}$ be the code with generating matrix $\Pm$, that is the set of elements of the form
$\xv \Pm$. Notice that the set of solutions of \eqref{eq:condition_square} is precisely 
the dual of $\sqc{\CC}$. This implies that $\sqc{\CC}$ should not be the whole space $\fq^n$.
This is quite unusual as explained in Section \ref{sec:basics} when the dimension $k$ of $\CC$ satisfies
$k \gg n^{1/2}$. Furthermore, since we are interested in solutions of 
\eqref{eq:condition_square} whose support is contained in $L$ we actually need
that the dual of $\sqc{\CC_L}$ is non empty which is even more abnormal since 
$\CC_L$ is a code of length $3\ell$ and dimension $\ell$. In other words, the Bogdanov and Lee homomorphic scheme 
leads in a natural way to choose codes $\CC$ which have a non-random behavior with respect to the dimension of the
square product.

\section{BBCRS Cryptosystem}
   \label{sec:schemeit}

\subsection{Description of the scheme}\label{ss:schemeit}
The cryptosystem denoted by BBCRS proposed by Baldi et \textit{al.}  in \cite{BBCRS11a} is a variant of the McEliece cryptosystem \cite{McEliece78} which replaces the permutation matrix used to hide the secret generator matrix by one of the form $ \Pim + \Rm$ where
$\Pim$ is a sum of $m$ permutation matrices and $\Rm$ is a matrix of rank $z$. 
Notice that the case $m=1$ and $z=0$ corresponds to the McEliece cryptosystem based on generalized Reed-Solomon codes (which was broken in \cite{SidelShesta92}).
Here we focus on the case where $z =1$, and $\Pim$ is a single permutation matrix which concerns all the parameters suggested in Section 5 of \cite{BBCRS11a}. There is actually a good reason
why the case $m=1,z=1$ stands out here:
$m=1$ is precisely the case  which gives by far the smallest key sizes when the parameters are chosen
so as to avoid generic decoding techniques aiming at recovering the message. 
Moreover, there is a big prize coming with increasing the value of $z$.
Basically the deciphering time is proportional to $q^z T$ where $q$ is the size of the field over which the
public code is defined (it is typically of the same order as the length $n$ of the code) and $T$ is the decoding time of the
GRS code used in this scheme. Roughly speaking, deciphering is about $n^z$ more complex than in a  McEliece cryptosystem
based on GRS codes. It was assumed in  \cite{BBCRS11a} that the gain in the public key size of the scheme would outweigh the
big loss in deciphering time. For this reason it is certainly questionable whether schemes with $z \geq 2$ could be really practical.
After the attack, which is detailed in this section, appeared on \texttt{www.arXiv.org} in \cite{Got12a}, a new version of
\cite{BBCRS11a} came out in \cite{BBCRS11b} where a slight generalization of $\Pim$ is considered,
namely $\Pim$ is just sparse now and the actual parameters proposed in \cite{BBCRS11b} 
suggest now matrices $\Pim$ with a row/column weight between $1$ and $2$. The 
attack proposed here does not apply directly to these new parameters anymore. It raises
the issue whether a generalization of our attack would be able to break the new
parameters, but this is beyond the scope of this paper.

From the authors' point of view, the idea underlying these new transformations was they would
allow to use  families of codes that were shown insecure in the
original McEliece cryptosystem. In particular, it would become possible to use GRS codes in this new framework.
The scheme can be summarized as follows.

\begin{description}
	\item \textbf{Secret key.} 
          \begin{itemize}
          \item $\Gms$ is a generator matrix of a GRS code of length $n$ and dimension $k$ over $\fq$,  
          \item $ \Qm  \eqdef \Pim + \Rm $ where
            $\Pim$ is an $n \times n$ permutation matrix;
          \item $\Rm$ is a rank-one matrix over $\fq$ such that $\Qm$
            is invertible. In other words there exist 
             $\alphav \eqdef (\alpha_1, \dots{}, \alpha_n) $ and $\betav \eqdef (\beta_1, \dots{}, \beta_n)$ in $\fq^n$
 such that  $\Rm \eqdef \alphav^T \betav$. 
          \item $\Sm$ is a $k \times k$ random invertible  matrix over $\fq$.
          \end{itemize}
        \item \textbf{Public key.} $\displaystyle \Gmp \eqdef \Sm^{-1} \Gms \Qm^{-1}$. 
          
	\item \textbf{Encryption.} The ciphertext $\cv \in \fq^n$ of a plaintext
          $\mv \in \fq^k$ is obtained by drawing at random $\ev$
          in $\fq^n$ of weight less than or equal  to $\frac{n-k}{2}$ and computing
          $\displaystyle \cv \eqdef \mv \Gmp  +  \ev$. 
          
	\item \textbf{Decryption.} It consists in performing the three
          following steps:
	\begin{enumerate}
	\item Guessing the value of  $\ev  \Rm$;
	\item Calculating $\cv' \eqdef \cv \Qm - \ev \Rm= \mv \Sm^{-1}\Gms + \ev  \Qm - \ev \Rm =  \mv \Sm^{-1}\Gms + \ev  \Pim $
	and using the decoding algorithm of the GRS code to recover
	$\mv \Sm^{-1}$ from the knowledge of $\cv'$;
	\item Multiplying the result of the decoding by $\Sm$ to recover $\mv$.
	\end{enumerate}
\end{description}

The first step of the decryption, that is guessing the value $\ev
\Rm$, boils down to trying $q$ elements (in the worst case) since
 $\ev \Rm = \ev \alphav^T \betav= \gamma \betav$ where $\gamma$
 is an element of $\fq$.

 \subsection{Key-recovery attack when  $2k+2 < n$}\label{sec:attack_Baldi}
We define $\Csec$ and $\Cpub$ to be the codes
generated by the matrices $\Gms$ and $\Gmp$ respectively. 
We denote by $n$ the length of these codes and by $k$ their 
dimension. We  assume in this subsection that
\begin{equation}
\label{eq:small_rate}
2k +2 < n
\end{equation}
The case of rates larger than $1/2$ will be treated in Subsection~\ref{subsec:dual_for_high_rates}.
As explained in Subsection \ref{ss:schemeit}, $\Csec$ is a GRS code. 
 It will be convenient to bring in the code
 \begin{equation}
   \label{eq:CC}
   \CC \eqdef \Csec \Pim^{-1}.
 \end{equation}
This code $\CC$, being a permutation of a GRS code, is itself a GRS 
code. So there are elements $\xv$ and $\yv$ in $\fq^n$ such that
$\displaystyle \CC = \GRS{k}{\xv}{\yv}$. 
There is a simple relation between $\Cpub$ and $\CC$ as explained
by Lemma~\ref{lem:structure} below.

First, notice that, since $\Rm$ has rank $1$, then so does $\Rm\Pim^{-1}$.
Hence there exist $\av$ and $\bv$ in $\fq^n$ such that:
\begin{equation}
\label{eq:RPi-1}
\Rm \Pim^{-1} = \bv^T \av.
\end{equation}

\begin{lem}\label{lem:structure}
Let $\lambdav \eqdef -\frac{1}{1+\scp{\av}{\bv}} \bv$.
For any
$\cv$ in $\Cpub$ there exists $\pv$ in $\CC$ such that: 
\begin{equation}
\label{eq:structure}
\cv = \pv + (\scp{\pv}{\lambdav}) \av.
\end{equation}
\end{lem}

\begin{proof}
  Appendix \ref{appendixBaldi}.
\end{proof}

\begin{rem}
  Notice that the definition of $\lambdav$ makes sense if and only of
  $\scp{\av}{\bv}\neq -1$. This actually holds since $\Qm$ is assumed to be
  invertible (See Lemmas~\ref{lem:PQinv} and \ref{lem:inverse}
  in Appendix~\ref{appendixBaldi}).
\end{rem}

From now on, we make the assumption that 
\begin{equation}
\label{eq:assumption}
\lambdav \notin \CC^\perp \ \textrm{and}\ \av \notin \CC.
\end{equation}
If this is not the case then $\Cpub=\CC=\GRS{k}{\xv}{\yv}$ and there is a
straightforward attack by applying the Sidelnikov and Shestakov algorithm
\cite{SidelShesta92} or the alternative attack we propose in Section~\ref{sec:schemeit}. It  finds $(\xv',\yv')$ that expresses $\Cpub$ as $\GRS{k}{\xv'}{\yv'}$. 
Our attack relies on identifying a code of dimension $k - 1$
that is both a subcode of  $\Cpub$ and the
GRS code $\CC$. It consists more precisely of
codewords $\pv + (\scp{\pv}{\lambdav}) \av$ with $\pv$ in $\CC$ such that  
$\scp{\pv}{\lambdav} = 0$. This particular code which is denoted by
$\CC_{\lambdav^\perp}$ is therefore:
\begin{equation}
  \label{eq:Clp}
\CC_{\lambdav^\perp} \eqdef \CC \cap <\lambdav>^\perp
\end{equation}
where $<\lambdav>$ denotes the vector space spanned by $\lambdav$. 
It is a subspace of $\Cpub$ of codimension $1$ if Assumption~(\ref{eq:assumption}) holds.
Here is an inclusion diagram for the involved codes. 
\begin{equation}
\label{eq:inc_diagram}
\xymatrix{\relax \Cpub  \ar@{-}[rd]_{\textrm{Codim} 1} & &
  \CC \ar@{-}[ld]^{\textrm{Codim} 1} \\
             & \CC_{\lambdav^\perp} & }
\end{equation}

\paragraph{Summary of the attack.}
Before describing it in depth, let us give the main steps of the attack.

\begin{enumerate}[\emph{Step} 1.]
\item Compute a basis of $\Clp$ using distinguisher-based methods.
   See \S~\ref{subsec:FirstStep} for further details.
\item Use Wieschebrink's method \cite{Wie10},
  which asserts that: $\Clp^2 = \CC^2$ to
  recover the structure of $\CC^2$ and then that of $\CC$. 
  See \S~\ref{subsec:SecondStep}.
\item Compute a pair $(\av_0, \lambdav_0)$
  called a \emph{valid pair} (Definition \ref{defn:valid}),
  which will have similar properties than the pair $(\av, \lambdav)$
  (see~(\ref{eq:RPi-1}) and Lemma~\ref{lem:structure} for the definitions
  of $\av$ and $\lambdav$). See \S~\ref{subsec:ThirdStep}.
\item Thanks to the valid pair, one can decrypt any ciphered message.
  See \S~\ref{subsec:FinalStep}.
\end{enumerate}

\subsubsection{Computing a basis of $\Clp$}
\label{subsec:FirstStep}
The inclusion relations described in the diagram (\ref{eq:inc_diagram})
strongly suggest that $\sqc{\Cpub}$ should have an unusual
low dimension since  $\sqc{\CC}$ has dimension  $2k-1$ by
Proposition \ref{prop:square}. 
More  exactly we have the following result.

\begin{prop}
\label{prop:dimension_sqCpub} The square code of $\Cpub$ satisfies
$\dim\left( \sqc{\Cpub} \right) \leqslant  3k-1$.
\end{prop}

\begin{proof}
To prove the result, let $(\bv_1, \dots , \bv_{k-1}, \bv_k)$ be a basis of $\Cpub$,
such that $(\bv_1, \dots , \bv_{k-1})$ is a basis of $\Clp$. Since $\Clp$ 
is a subcode of the GRS code $\CC$, which is of dimension $k$, we have then $\dim\left(\sqc{\Clp} \right) \leq 2k - 1$, 
and the vectors $\bv_i \cwp \bv_j$ with $1\leq i,j\leq k$ generate $\sqc{\Cpub}$.
Among these vectors only $\bv_k \cwp \bv_i = \bv_i \cwp \bv_k$ for $1 \leq i \leq k$ are possibly not in $\sqc{\Clp}$. Therefore, 
$\dim\left (\sqc{\Cpub} \right ) \leq 2k-1 + k $.
\end{proof}

\begin{rem}
Experimentally it has been observed that the upper-bound is sharp.
Indeed, the dimension of $\sqc{\Cpub}$ has always been found
to be equal to $3k-1$ in all our experiments when choosing randomly the
codes and $\Qm$ with parameters of \cite{BBCRS11a} of Example~1 and~2. 
In our tests we randomly picked $1000$ GRS codes with rate $\leq 1/2$, apply random transformations $\Qm^{-1}$ on them. 
\end{rem}

The second observation is that when a basis 
$\gv_1,\dots,\gv_k$ of $\Cpub$ is chosen together with $l$ other random elements $\zv_1,\dots,\zv_l \in  \Cpub$, then we may expect that 
the dimension of the vector space generated by all products $\zv_i \cwp \gv_j$ with $i$ in $\{1,\dots,l\}$ and $j$ in 
$\{1,\dots,k\}$ is the dimension of the full space $\sqc{\Cpub}$ when
$l \geqslant 3$. This is indeed the case when $l \geqslant 4$ 
but it is not true for $l=3$ since we have the following result.

\begin{prop}\label{prop:three} Let $\code{B}$ be the linear space spanned by $\Big
  \{ \zv_i \cwp \gv_j ~|~ 1 \leqslant i \leqslant 3 \text{ and } 1 \leqslant j
  \leqslant k \Big \}$ then it holds:
  $$
  \dim \left ( \code{B}   \right ) \leqslant 3k-3.
  $$
\end{prop}
A proof of this phenomenon is given in Appendix~\ref{proof:dim3k_3}. 
Experimentally, it turns out that almost always this upper-bound
is tight and the dimension is generally $3k-3$. But if we assume
now that $\zv_1$, $\zv_2$, $\zv_3$ all belong to
$\CC_{\lambdav^\perp}$, which happens with probability $\frac{1}{q^3}$ since $\CC_{\lambdav^\perp}$ is 
a subspace of $\Cpub$ of codimension $1$ (at least when~(\ref{eq:assumption}) holds),
then the vectors $\zv_i \cwp \gv_j$ generate a subspace with  a much
smaller dimension.
\begin{prop}\label{prop:attack}
If $\zv_i$ is in $\CC_{\lambdav^\perp}$ for $i$ in $\{1,2,3\}$
then for all $j$ in 
$\{1,\dots,k\}$:
\begin{eqnarray}
 \zv_i \cwp \gv_j    ~\subset~  \sqc{\CC} ~+~ <\zv_1 \cwp \av> ~+~
<\zv_2 \cwp \av> ~+~ <\zv_3 \cwp \av> \label{eq:idea}
\end{eqnarray}
and if $\code{B}$ is the linear code spanned by $\big
  \{ \zv_i \cwp \gv_j ~|~ 1 \leqslant i \leqslant 3 \text{~and~} 1 \leqslant j
  \leqslant k \big \}$ then
\begin{eqnarray}
 \dim \left( \code{B} \right ) \leqslant 2k+2.  \label{eq:consequence}
\end{eqnarray}
\end{prop}

\begin{proof}
Assume that the $\zv_i$'s all belong to $\CC_{\lambdav^\perp}$. For every $\gv_j$ there exists
$\pv_j$ in $\CC$ such that $\gv_j=\pv_j + \scp{\lambdav}{\pv_j}\av$.
We obtain now
\begin{eqnarray}
\zv_i \cwp \gv_j & = & \zv_i \cwp (\pv_j + (\scp{\lambdav}{\pv_j})\av) \nonumber \\
& =& \zv_i \cwp \pv_j + (\scp{\lambdav}{\pv_j})\zv_i \cwp \av \nonumber \\
& \in & \sqc{\CC} + <\zv_1 \cwp \av> + <\zv_2 \cwp \av> + <\zv_3 \cwp \av>.
\end{eqnarray}
This proves the first part of the proposition, the second part follows immediately from 
the first part since it implies that the dimension of
the vector space generated by the $\zv_i \cwp \gv_j$'s is upperbounded by the
sum of the dimension of $\sqc{\CC}$ (that is $2k-1$) and the dimension of the 
vector space spanned by the $\zv_i \cwp \av$'s (which is at most $3$).
\end{proof}

The upper-bound 
given in \eqref{eq:consequence} on the dimension
follows immediately from \eqref{eq:idea}. This leads to Algorithm \ref{algo:Clambdaperp} which computes
a basis of $\CC_{\lambdav^\perp}$. It is essential that the condition
in \eqref{eq:small_rate} holds in order to distinguish the case when the 
dimension is less than or equal to $2k+2$ from higher dimensions.
\begin{algorithm}[t]
  {\bf Input: } A basis $\{\gv_1,\dots,\gv_k\}$ of $\Cpub$.\\
  {\bf Output : } A basis $\mathcal{L}$ of $\CC_{\lambdav^\perp}$.

  \begin{algorithmic}[1]
  \REPEAT \label{re}
   \FOR{$1 \leqslant  i \leqslant  3$}
   \STATE{Randomly choose $\zv_i$ in $\Cpub$}
   \ENDFOR
   \STATE{ $\code{B} \leftarrow ~ < \big\{ \zv_i \cwp \gv_j ~|~ 1 \leqslant i \leqslant 3 \text{~and~} 1 \leqslant j
  \leqslant k \big \} >$}
  \UNTIL{$\dim(\code{B}) \leqslant  2k+2$ and $\dim
    \left(<\zv_1,\zv_2,\zv_3> \right) = 3$}
  \STATE{$\mathcal{L} \leftarrow \{\zv_1,\zv_2,\zv_3\}$}
  \STATE{$s \leftarrow 4$}
  \WHILE{$s \leqslant  k-1$}
    \REPEAT
  \STATE{Randomly choose $\zv_s$ in $\Cpub$}   
  \STATE{$\code{T} \leftarrow ~ < \big\{ \zv_i \cwp \gv_j ~|~ i \in \{1,2,s\} \text{~and~} 1 \leqslant j
  \leqslant k \big \} >$}
  \UNTIL{$\dim(\code{T}) \leqslant  2k+2$ \AND $\dim \left(< \mathcal{L} \cup \left\{
      \zv_s \right \} > \right) = s $}
  \STATE{$\mathcal{L} \leftarrow \mathcal{L} \cup \{\zv_s\}$}
  \STATE{$s \leftarrow s+1$}
  \ENDWHILE
    \RETURN{$\mathcal{L}$;}
  \end{algorithmic}
  \caption{\label{algo:Clambdaperp}Recovering $\CC_{\lambda^\perp}$.}
\end{algorithm}
The first phase of the attack, namely finding a suitable triple $\zv_1,\zv_2,\zv_3$ runs 
in expected time $O\left( q^3 k^2 n \right)$
because 
each test in the \textbf{repeat} loop \ref{re} has a chance of
$\frac{1}{q^3}$ to succeed. Indeed, 
$\CC_{\lambdav^\perp}$ is of codimension $1$ in $\Cpub$ and therefore
a fraction $\frac{1}{q}$ of elements of $\Cpub$ belongs to
$\CC_{\lambdav^\perp}$.
Once $\zv_1,\zv_2,\zv_3$ are found, getting any other element of
$\CC_{\lambdav^\perp}$ is easy. Indeed, take a random element $\zv \in \Cpub$
and use the same test to check whether the triple $\zv_1, \zv_2, \zv$
is in $\CC_{\lambdav^\perp}$. Since $\zv_1, \zv_2 \in
\CC_{\lambdav^\perp}$ the probability of success is $\frac{1}{q}$ and
hence $\zv$ can be found in $O(q)$ tests. 
The whole algorithm runs in expected time
$O\left( q^3 k^2 n \right)+ O\left( q k^3 n \right)= O\left( q^3 k^2 n \right)$
since $k < n \le q$, hence the first phase of the attack  
is dominant in the complexity.

\subsubsection{Recovering the structure of $\CC$}
\label{subsec:SecondStep}
Once $\CC_{\lambdav^\perp}$ is
recovered, it still remains to recover the secret code and $\av$. 
The problem at hand can be formulated like this: we know a very large subcode,
namely $\CC_{\lambdav^\perp}$,
of a GRS code that we want to recover.
This is exactly the problem which was solved in \cite{Wie10}.
In our case this amounts to compute  $\sqc{\CC_{\lambdav^\perp}}$ which turns out
to be equal to $\GRS{2k-1}{\xv}{\yv\cwp\yv}$ (see \cite{MMP11a,MMP11a:DCC}
for more details). 
It suffices to use the Sidelnikov and Shestakov algorithm
\cite{SidelShesta92} or the algorithm described in Section~\ref{sec:GRS_attack}
to compute a pair $(\xv,\yv \cwp \yv)$ describing $\sqc{\CC_{\lambdav^\perp}}$
as a GRS code.  From this, we deduce a pair $(\xv,\yv)$ defining 
the secret code $\CC$ as a GRS code.

\subsubsection{Deriving $\av$ and $\lambdav$ from $\CC$ and $\CC_{\lambdav^\perp}$}
\label{subsec:ThirdStep}

At this step of the attack let us summarize what has been done.
We have been able to compute the codes $\CC$ and $\Clp$
defined in~(\ref{eq:CC}) and~(\ref{eq:Clp}) respectively.
We recall the inclusion diagram.
$$
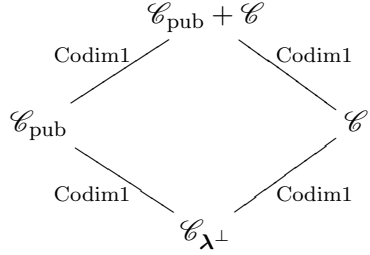
\xymatrix{\relax  & \Cpub + \CC \ar@{-}[ld]_{\textrm{Codim} 1} \ar@{-}[rd]^{\textrm{Codim} 1}& \\ 
  \Cpub  \ar@{-}[rd]_{\textrm{Codim} 1} & &
  \CC \ar@{-}[ld]^{\textrm{Codim} 1} \\
             & \CC_{\lambdav^\perp} & }
$$
In addition, we know that the code $\CC$ and $\Cpub$
are related by the map
\begin{equation}
  \label{eq:mappsi}
\psi_{ \av, \lambdav } :
\left\{
  \begin{array}{rcl}
    \CC & \rightarrow & \Cpub\\
    \pv & \mapsto     & \pv + (\scp{\pv}{\lambdav})\av 
  \end{array}
\right. .
\end{equation}

To finish the attack, we need to find a pair
$(\av_0, \lambdav_0)\in \fq^n \times \fq^n$
such that the map $\psi_{ \av_0 , \lambdav_0 }$ induces an isomorphism
from $\CC$ to $\Cpub$. This motivates the following definition.
\begin{defn}
  \label{defn:valid}
  A pair $(\av_0, \lambdav_0)\in \fq^n \times \fq^n$
  is said to be a \emph{valid pair} if
  \begin{enumerate}[(a)]
  \item\label{item:scpal} $\scp{\av_0}{\lambdav_0} \neq -1$;
  \item\label{item:psiCCincCpub} $\psi_{\av_0, \lambdav_0} (\CC) \subseteq \Cpub$.
  \end{enumerate}
\end{defn}

\begin{rem}
  \label{rem:inciseq}
  From Corollary~\ref{cor:inverse} (Appendix~\ref{appendixBaldi}),
  Condition~(\ref{item:scpal})
  asserts that $\psi_{\av_0, \lambdav_0}$
  is an isomorphism.
  Thus, 
  $$
  \forall \pv \in \Cpub,\ \exists \pv' \in \CC,\ {\rm such}\ {\rm that}\
  \pv = \pv' + (\scp{\pv'}{\lambdav}_0)\av_0.
  $$
  Moreover, if~(\ref{item:scpal}) holds then
  the inclusion in~(\ref{item:psiCCincCpub}) is an equality
  since both codes have the same dimension.
\end{rem}

\bigskip

First, we choose $\uv \in \CC\setminus \Clp$ and
$\vv \in \Cpub \setminus \Clp$.
Since $\Clp$ has codimension $1$ in $\CC$, we have
\begin{equation}
  \label{eq:decompCC}
  \CC = \Clp \oplus <\uv> \quad \textrm{and}\quad \Cpub = \Clp \oplus <\vv>.
\end{equation}

\noindent A valid pair $(\av_0, \lambdav_0)$ can be found easily
using the two following elementary lemmas.

\begin{lem}
  \label{lem:lv0.u}
  For all $\lambdav_0 \in \Clpp \setminus (\CC^{\perp} \cup \Cpub^{\perp})$,
  we have 
  $$
  \scp{\lambdav_0}{\uv} \neq 0 \quad \textrm{and} \quad
  \scp{\lambdav_0}{\vv}\neq 0.
  $$
\end{lem}

\begin{proof}
  Assume that $\scp{\lambdav_0}{\uv} = 0$. Then, 
  $\lambdav_0 \in \Clpp \cap <\uv>^{\perp} = (\Clp + <\uv>)^{\perp}$.
  Hence, from~(\ref{eq:decompCC}), we would have $\lambdav_0 \in \CC^{\perp}$
  which yields a contradiction.
  The other non-equality is proved by the very same manner.
\end{proof}

\begin{lem}
  \label{lem:psiuinCpub}
  For all $\lambdav_0 \in \Clpp$ and for all
  $\xv \in \fq^n$, we have
  $$
  \psi_{\lambdav_0, \xv} (\CC) \subset \Cpub \ \Longleftrightarrow 
  \psi_{\lambdav_0, \xv} (\uv) \in \Cpub.
  $$
\end{lem}

\begin{proof}
  Since $\uv \in \CC$, the implication $(\Longrightarrow)$ is obvious.
  Conversely, assume that $\psi_{\lambdav_0, \xv} (\uv) \in \Cpub$.
  Then, from~(\ref{eq:decompCC}), to show the result there remains
  to show that $\psi_{\lambdav_0 ,  \xv} (\Clp) \subset \Cpub$.
  But, since $\lambdav_0 \in \Clpp$, then for all
  $\pv \in \Clp$, we have
  $$\psi_{\lambdav_0 ,  \xv} (\pv)= \pv + (\scp{\lambdav_0}{\pv})
  \xv = \pv.$$
  Thus, $\psi_{\lambdav_0 ,  \xv} (\Clp) = \Clp \subset \Cpub$.
\end{proof}

\paragraph{Procedure to recover a valid pair.}
Before starting, recall that we fixed vectors $\uv \in \CC \setminus \Clp$
and $\vv \in \Cpub \setminus \Clp$
so that~(\ref{eq:decompCC}) holds.
\begin{enumerate}[{\em Step} 1.]
\item Choose $\lambdav_0 \in \Clpp \setminus (\CC^{\perp} \cup \Cpub^{\perp})$
  at random.
  Notice that the set $\Clpp \setminus (\CC^{\perp} \cup \Cpub^{\perp})$
  is nonempty since both $\CC^{\perp}$ and $\Cpub^{\perp}$ have codimension $1$
  in $\Clpp$ and even over a finite field, no vector space of dimension $\geq 1$
  is a union of two vector subspaces of codimension $1$.
\item Set
  $$
  \av_0:= \frac{1}{\scp{\lambdav_0}{\uv}} \left(\vv - \uv \right).
  $$
  It is well--defined thanks to Lemma~\ref{lem:lv0.u}.
\end{enumerate}
  We claim that the pair $(\av_0, \lambdav_0)$ is valid.
  Indeed, we have
  $$
    \scp{\av_0}{\lambdav_0} =
    \frac{\scp{\lambdav_0}{\vv}}{\scp{\lambdav_0}{\uv}} - 1.
  $$
  Moreover, $\scp{\lambdav_0}{\vv}\neq 0$ thanks
  to Lemma~\ref{lem:lv0.u}, and hence $\scp{\av_0}{\lambdav_0}\neq -1$.
  Thus, the pair satisfies Condition~(\ref{item:scpal}) of
  Definition~\ref{defn:valid}.

  To show that Condition~(\ref{item:psiCCincCpub}) is satisfied too,
  Lemma~\ref{lem:psiuinCpub} asserts that we only need to
  prove that $\psi_{\av_0, \lambdav_0}(\uv) \in \Cpub$ which is true since 
  an elementary computation yields
  $$
  \psi_{\av_0, \lambdav_0}(\uv) = \vv
  $$
  which is in $\Cpub$ by construction.

\subsubsection{Decryption of any ciphertext}
\label{subsec:FinalStep}
We have found a valid pair (Definition~\ref{defn:valid})
$(\av_0,\lambdav_0)$.
We want to decode the vector $\zv \eqdef \cv +\ev$ where $\ev$ is an error of a certain Hamming weight which 
can be corrected by the decoding algorithm chosen for $\CC$ and $\cv$ is an element of the public code. From Remark~\ref{rem:inciseq}
page~\pageref{rem:inciseq},
we know that there exists 
$\pv$ in $\CC$ such that 
\begin{equation}
\label{eq:cv_pv}
\cv = \pv + (\scp{\lambdav_0}{\pv}) \av_0.
\end{equation}
We compute $\zv(\alpha) \eqdef \zv + \alpha \av_0$ for all elements 
$\alpha$ in $\fq$. One of these elements $\alpha$ is equal to $-\scp{\lambdav_0}{\pv}$  and 
we obtain $\zv(\alpha) = \pv + \ev$ in this case. Decoding $\zv(\alpha)$ in $\CC$ will reveal $\pv$ and
this gives $\cv$ by using Equation~\eqref{eq:cv_pv}.

\subsection{Extending  the attack  for rates larger than $\frac{1}{2}$}
\label{subsec:dual_for_high_rates}

The codes suggested in \cite[\S5.1.1,\S5.1.2]{BBCRS11a} are all of rate significantly larger than $\frac{1}{2}$,
for instance Example 1 p.15 suggests a GRS code of length $306$, dimension $232$ over $\ff{307}$, 
whereas Example 2. p.15 suggests a GRS code of length $511$, dimension $387$ over $\ff{512}$.
The attack suggested in the previous subsection only applies to rates smaller than $\frac{1}{2}$.
There is a simple way to adapt the previous attack for this case by
considering the dual  $\Cpub^\perp$ of the public code. Note that by
Proposition \ref{pr:dual},  there exists $\yv'$ in $\fq^n$ for which
we have $\displaystyle \CC^\perp = \GRS{n-k}{\xv}{\yv'}$. Moreover,
$\Cpub^\perp$ displays a similar structure as $\Cpub$.

\begin{lem}\label{lem:structure_dual}
For any $\cv$ from $\Cpub^\perp$ there exists an element $\pv$ in $\CC^\perp$ such that:
\begin{equation}
\label{eq:structure_dual}
\cv = \pv + (\scp{\pv}{\av}) \bv.
\end{equation}
\end{lem}

\begin{proof}
The key to Lemma \ref{lem:structure_dual} is the fact that,
from~(\ref{eq:Cpub=CsecPim1}), we have $\Cpub^{\perp} =
\CC^\perp \Pm^T$. Indeed $\Cpub = \CC\Pm^{-1}$ and therefore for any element
$\cv$ of $\Cpub$ there exists an  element $\pv$ of $\CC$ such that 
$\cv = \pv \Pm^{-1}$. Observe now that every element $\cv'$ in $\Cpub^\perp$
satisfies
$$
0=\scp{\cv}{\cv'} = \scp{\pv \Pm^{-1}}{\cv'}.
$$
If we set $\cv' = \pv' \Pm^{T}$ it results $\scp{\pv}{\pv'} = 0$, therefore
$\Cpub^\perp = \CC^\perp \Pm^T$.
This discussion implies that
there exists an element $\pv'$ in $\CC^\perp$ such that:
$$
\cv'  =  \pv' \Pm^T
 =  \pv'  \left(\Im +  \bv^T \av \right)^T
 =  \pv'  + \pv' \av^T \bv 
 =  \pv' + (\scp{\pv' }{\av}) \bv.
$$
\end{proof}

It implies that the whole approach of the previous subsection can be 
carried out over $\Cpub^\perp$. It allows to recover the secret code $\CC^\perp$ and therefore also $\CC$. This attack needs
that $2(n-k)+2 < n$, that is $2k > n+2$. In summary, there is an attack as soon as $k$ is outside a narrow interval around
$n/2$ which is $[\frac{n-2}{2},\frac{n+2}{2}]$.

\section{McEliece Variants Based on GRS codes}
  \label{sec:GRS_attack}

In this section, we will give an alternative attack of \cite{SidelShesta92} against any McEliece-like cryptosystem based on GRS codes. 
This attack runs in polynomial time and makes possible the recovery of the
structure of any GRS code. From the computational point of view,
this attack is less efficient than that of Sidelnikov and Shestakov because
of the cost of the computation of squares or star products of
codes. 
Indeed,  the complexity of the Sidelnikov-Shestakov attack is $O\left(k^3+k^2n\right)$ whereas our attack runs in
 $O(k^2n^3+k^3n^2)$ operations. 
However, our approach remains of interest since it does not require as a
first step the computation of minimum weight codewords.
For this reason, it could provide interesting generalizations.
Indeed, it should be noticed that certain key recovery attacks
on other algebraic codes such as
\cite{MinderShokrollahi07} on Reed--Muller codes and \cite{FM08}
on hyperelliptic algebraic geometry codes are built in the same spirit as 
Sidelnikov and Shestakov's attack \cite{SidelShesta92} and in particular have as a first
step, the computation of minimum weight codewords.
This computation is subexponential for Reed--Muller codes and
exponential in the genus of the curve for algebraic geometry codes, which
limits the attack \cite{FM08} to codes from curves with very low genus.
On the other hand, our method might be generalized to such codes
and provides alternative and more computationally efficient attacks.

\subsection{Context and notation}
 Let $\CC$ be a $q$-ary GRS code $\CC \eqdef \GRS{k}{\av}{\bv} \subset \fq^n$.
 Assume that it has dimension $k\leq n/2$
 (if not, then one can work with the dual code). 
First assume that the two first positions, \textit{i.e.} the two first
entries of $\av$ are $0$ and $1$.
Such an assumption makes sense since every GRS code
is permutation equivalent to a code satisfying this condition.
This is a consequence of the $3$--transitivity
of the action of the projective linear group $\mathbf{PGL} (2, \fq)$
on the points of the projective line.

\begin{nota}
For all $i,j$ such that $i>0$, $j>0$ and $i+j\leq k-1$, we denote by
$\CC (i,j)$ the subcode of $\CC$ given by the evaluation of
polynomials vanishing at $0$ (\textit{i.e.} the first position by assumption) with
multiplicity at least $i$ and at $1$ (\textit{i.e.} the second position)
 with multiplicity at least $j$, \textit{i.e.} multiples of $x^i(x-1)^j$.   
For convenience sake, we set $\CC (0,0) \eqdef \CC$.
\end{nota}

The main step of our attack is to compute some codes among
$\CC(i,j)$.
Notice that these codes are also GRS codes.

\subsection{Computing some subcodes}
Clearly, the computation of a generator matrix of $\CC (0,1), \CC (1,0)$ and
$\CC(1,1)$ is straightforward since it reduces to Gaussian elimination.
These codes are nothing but shortenings of $\CC$.
The main tool of our attack is the following result.

\begin{thm}
  Assume that $k\leq n/2$.
  For all $1\leq i \leq k-2$ and all $j$ such that $i+j\leq k-2$, we have
$$
  \CC (i+1, j) \cwp \CC (i-1 , j) = \CC (i,j)^2 \quad \textrm{and} \quad
  \CC (i, j+1) \cwp \CC (i , j-1) = \CC (i,j)^2.
$$
\end{thm}

\begin{proof}
  We prove the first identity, the second is obtained easily by symmetry.
  For all pair of nonnegative integers $(i,j)$, set
    $$
        V_{i,j}   \eqdef   x^i (x-1)^j \fq [x]_{<k-i-j}
    $$
  This space has dimension $k-i-j$
  and is related to our GRS codes by
  $$
        \CC(i,j)       =  < \bv \cwp P(\av) \ |\ P\in V_{i,j}>,
  $$
  where for all $P\in \fq[x]$, we denote by $P(\av)$ the word
$P(\av)\eqdef (P(a_1), \ldots , P(a_n))$. Clearly, we have:
  $$
    V_{i,j}^2  = x^{2i}(x-1)^{2j} \fq[x]_{<2k-2i-2j-1}
  $$
  and it is also readily  checked that
  $$
   V_{i-1, j} \cwp V_{i+1,j}  = x^{2i}(x-1)^{2j} \fq[x]_{<2k-2i-2j-1}.
  $$
  This yields the result.
\end{proof}

From the previous result, as long as $\CC (i,j)^2 \neq \fq^n$,
which holds for $k\leq n/2$, given generator matrices of
$\CC(i,j)$ and $\CC(i-1, j)$, one can recover a basis of
$\CC (i+1, j)$ by solving a simple linear system.
Indeed, deciding whether an element $\cv \in \CC(i,j)$ is actually
in $\CC(i+1, j)$ reduces to solve:
\begin{equation}
  \cv \cwp \CC(i-1, j) \subseteq \CC (i, j)^2. \label{eq:filtration}
\end{equation}
It is worthwhile noting that \eqref{eq:filtration} is not satisfied for a $\cv \in \CC(i,j)$ that does not belong to $\CC(i+1, j)$.

\paragraph{Complexity.}
To solve \eqref{eq:filtration}, we first need to compute a row-echelon
basis for $\CC(i,j)^2$.
From Proposition~\ref{prop:complexity}, this costs $O(k^2n^2)$.
From this basis, we compute easily a basis for
${\left(\CC(i,j)^2\right)}^{\bot}$.
The equations of the linear system~(\ref{eq:filtration}) have the form $\hv \cwp \dv$ where 
$\hv \in \CC(i-1,j)$ and $\dv \in (\CC(i,j)^2)^{\bot}$.
Thus, solving the system consists in computing all these equations whose number 
is $(\dim \CC(i-1,j))(n-\dim \CC(i,j)^2)$. Hence their computation costs
$O(kn(n-2k))$, then we solve a linear system which costs $O(n^2k(n-2k))$,
or roughly speaking $O(n^3k)$.
Therefore, the complexity of solving (\ref{eq:filtration}) is
$O(k^2n^2+k^3n)$.
This computation should be iterated $k$ times, which yields 
$O(k^2n^3+k^3n^2)$ operations. 

\subsection{Description of the attack}
The attack summarizes as follows. We assume that the dimension of the GRS
code is less than $n/2$, if not one can apply the attack on its dual.

\begin{enumerate}[{\it Step 1.}]
  \item Compute a basis of $\CC (k-1 , 0)$, \textit{i.e.} compute a nonzero vector $\cv$
    of this $1$--dimensional space.
    The corresponding vector comes from the evaluation of
    a polynomial of the form $\lambda x^{k-1}$ for some
    $\lambda \in \fq^{\times}$.
    More precisely, we get the vector $\lambda (\av^{ k-1} \cwp \bv)$.
    Then, compute a basis of $\CC (k-2, 1)$. The corresponding vector
    $\cv'$ is of the form $\mu \av^{k-2} \cwp (\av - \word{1}) \cwp \bv $ for $\mu \in \fq^{\times}$ and where $\word{1} \eqdef (1, \ldots , 1)$.

  \item The vectors $\cv$ and $\cv'$ have no zero position but the
    two first ones. Thus, after puncturing at the two first positions
    the quotient $\cv'/\cv$ makes sense and corresponds to the evaluation of the
    fraction $\nu (x-1)/x$ for some $\nu \in \fq^{\times}$
    (\textit{i.e.} is $\nu (\av - \word{1})/\av$, which makes sense after a suitable
    puncturing).

    It is worth noting that compared to the vectors $\cv$ and $\cv'$,
    the vector $\cv'/\cv$ corresponds to the exact evaluation of $\nu (x-1)/x$ at 
    some elements of $\fq \setminus \{0,1\}$ since 
    the entries of $\bv$ are cancelled by the quotient.

  \item Up to now, we only made two arbitrary choices by fixing the position 
    of $0$ and $1$. Because of the $3$--transitivity of $\mathbf{PGL}(2, \fq)$,
    one can make a third arbitrary choice. 
    Thus, without loss of generality, one can assume that $\nu =1$.
    Now, notice that the map $x \mapsto (x-1)/x$ is a bijection from
    $\fq \setminus \{0,1\}$ to itself with reciprocal map $y\mapsto 1/(1-y)$.

    Thus, by applying the map $y\mapsto 1/(1-y)$ to the entries of the vector
    $\cv'/\cv$ we get the corresponding positions, \textit{i.e.} the vector $\av$.

  \item
    Now, comparing the vector $c$ with the vector
    $\av^{k}$,
    we get $\bv$ up to 
    multiplication by an element $\alpha \in \fq^{\times}$, which does not matter
    since $\GRS{k}{\av}{\bv} = \GRS{k}{\av}{\alpha \bv}$ for
    $\alpha \in \fq^{\times}$. 
\end{enumerate}

\begin{rem}
  Roughly speaking, this attack can be regarded as a ``local version''
of Sidelnikov and Shestakov's attack. Indeed, Sidelnikov and Shestakov's
attack consist
in finding two codewords of minimum weight whose support differ only in two positions.
 This corresponds to shorten the code at $k-2$ positions and then
recover the structure of the code using two codewords of this shortened code.
Here, we shorten only in a single  position but consider polynomials vanishing with
a high multiplicity.
\end{rem}

\section*{Conclusion} 

In this paper we use directly the fact that the square of codes which are close enough to GRS codes 
have an abnormally small dimension. When applied to several public-key encryption schemes \cite{Nie86,Wie06,BBCRS11a,BL12},  it always results in an efficient key-recovery attack. More precisely, we show that:

\begin{itemize}
\item  Computing the dimensions of the square of various subcodes of the public code permits to detect random columns in the generator
matrix of the public code of the Wieschebrink cryptosystem \cite{Wie06},

\item  Computing the dimensions of the square of various punctured versions of the public code in the Bogdanov-Lee cryptosystem \cite{BL12} enables 
to retrieve the Reed-Solomon part of the public code,

\item In the case of the scheme \cite{BBCRS11a}, it is possible to identify a certain subcode that is both included 
in a GRS code and the public code,

\item In the case of a McEliece-like cryptosystem based on a GRS code \cite{Nie86},
it enables to get a full filtration by means of GRS subcodes, so that 
the structure of the public code as a GRS code is recovered.
\end{itemize}

It should be mentioned that the idea of using product codes and a suitable filtration was  used  recently in \cite{COT:EC14} 
to cryptanalyze successfully in polynomial time the wild McEliece cryptosystems proposed in \cite{BLP10:wild} that were 
 defined over a quadratic extension.

\medskip

 Note that the component-wise product of codes which is central to our approach has been applied recently in \cite{CB13a} to attack the McEliece variant 
 based on Reed-Muller codes proposed in \cite{Sid94}. The squares of these codes have also an abnormal dimension in this case.
 This yields in some cases a polynomial time attack \cite{CB13a} and in general it improves  upon the subexponential attack of \cite{MinderShokrollahi07}. 
 It would be interesting to study whether an attack similar to our filtration attack which was effective against GRS codes could be carried  out for Reed-Muller codes  to yield 
 a polynomial time attack on all instances of this cryptosystem.  However, the most 
challenging task would be to attack the original McEliece cryptosystem with similar tools (at least for a range of parameters) since
duals of Goppa codes also have, in a limited way, square codes with low dimensions.\footnote{See  \cite{MP12a} which contains much more examples of codes with this kind of behavior}

\bibliographystyle{alpha}
\bibliography{crypto}

\cleardoublepage
\newpage

\appendix

\section{Proof of Proposition~\ref{prop:dsg}}
 \label{sec:proof_prop:dsg}
 
Set $a \eqdef |I|-|J|$ and $b\eqdef |I|$.
After a suitable permutation of the support and the indexes of the $x_j$'s,
the code $\CC_I$ has a generator
matrix of the form
$$
\begin{pmatrix}
x_1 & x_2 & \cdots & x_a & x_{a+1} & \cdots & x_b \\
 \vdots & \vdots &    & \vdots & \vdots &        & \vdots \\
x_1^{\ell} & x_2^{\ell} & \cdots & x_a^{\ell} & x_{a+1}^{\ell} & \cdots &
x_b^{\ell}\\
 & & & & & & \\
x_1^{\ell +1} & x_2^{\ell +1} & \cdots & x_a^{\ell +1} &  &  & \\
 \vdots    &   \vdots & & \vdots &  & (0)  & \\
x_1^k & x_2^k & \cdots & x_a^k & &  & 
\end{pmatrix}
$$
We define the maps
$$
\Phi_{I} : \map{\fq [x]}{\fq^b}{P}{(P(x_1), \ldots , P(x_b))}
\quad \text{ and } \quad
\Phi_{I\setminus J} : \map{\fq[x]}{\fq^b}{P}{(P(x_1), \ldots , P(x_a), 0\ldots , 0)}.
$$
We have the two following obvious lemmas.
\begin{lem}
  \label{lem:injection}
Both maps $\Phi_{I}$ and $\Phi_{I\setminus J}$ are linear. In addition,
their restrictions to the vector space $<x^2, \ldots , x^{2k}>$ are injective.
\end{lem}

\begin{proof}
  It is sufficient to prove that the restriction of $\Phi_{I\setminus J}$ is
  injective.
  It is an elementary consequence of polynomial interpolation, since 
  $a = |I| - |J|$ is assumed to be be larger than $2k$.
\end{proof}

\begin{lem}
For all $P, Q\in \fq[x]$, we have:
\begin{eqnarray}
\label{haha}   \pI{P}\cwp \pI{Q} &  = &  \pI{PQ}\\
\label{hehe}   \pIJ{P}\cwp \pIJ{Q} &  = &  \pIJ{PQ}\\
\label{hoho}   \pI{P} \cwp \pIJ {Q} & = & \pIJ {PQ}
\end{eqnarray}  
\end{lem}

\noindent Clearly, we have
\begin{equation}
\CC_I = \pI{<x,\ldots, x^{^\ell}>}\ \oplus \ \pIJ{<x^{\ell +1},\ldots, x^k>}.
\end{equation}
Using (\ref{haha}), (\ref{hehe}) and (\ref{hoho}), we get
\begin{align*}
\CC_I^2 & =     \pI{<x, \ldots, x^{\ell}>}^2 \ + \
                  \pIJ{<x^{\ell +1}, \ldots , x^k>}^2\\
   & \qquad \qquad \qquad \qquad \qquad \qquad \qquad 
       \ + \ \pI{<x, \ldots, x^{\ell}>}\cwp \pIJ{<x^{\ell +1}, \ldots, x^k>} \\
        & =    \pI{<x^2, \ldots , x^{2\ell}>} \ + \
        \pIJ{<x^{2\ell+2}, \ldots, x^{2k}>}
        \ +\ \pIJ{<x^{\ell +2 }, \ldots, x^{k+ \ell}>}\\
        & =    \pI{<x^2, \ldots , x^{2\ell}>} \ + \
        \pIJ{<x^{2\ell+2}, \ldots, x^{2k}>
        \ +\ <x^{\ell +2 }, \ldots, x^{k+ \ell}>}
\end{align*}
Since, by assumption, $\ell <k$, we have
$$
<x^{\ell +2 }, \ldots, x^{k+ \ell}>
        \ +\ <x^{2\ell+2}, \ldots, x^{2k}>
 \ =\ <x^{\ell+2}, \ldots, x^{2k}>
$$
Therefore,
\begin{equation}
\label{eq:dimCCIsquare}
\CC_I^2 = \pI{<x^2, \ldots, x^{2\ell}>} \ +\ \pIJ{<x^{\ell+2}, \ldots, x^{2k}>}.
\end{equation}
Lemma~\ref{lem:injection} entails
\begin{equation}
  \label{eq:dims}
  \dim \pI{<x^2, \ldots , x^{2\ell}>} = 2\ell -1,\ \ \textrm{and}\ \ 
  \dim \pIJ{<x^{\ell+2}, \ldots , x^{2k}>} = 2k-\ell -1.
\end{equation}
To conclude the proof, we need to compute the dimension of the intersection of these spaces.
For this purpose, set
$$
R(x)\eqdef \prod_{j=a+1}^b (x-x_j).
$$
An element of $\pI{<x^2, \ldots , x^{2\ell}>} \cap
\pIJ{<x^{\ell+2},\ldots , x^{2k}>}$ is
an element of $\pI{<x^2 , \ldots, x^{2\ell}>}$ which vanishes on the
$|J| = b-a$ last
positions: it is an element of 
$\pI{<x^2 R(x), \ldots , x^{2\ell-|J|}R(x)>}$.
Thus,
\begin{align*}
\pI{<x^2, \ldots, x^{2 \ell}>} \cap &\ \pIJ{<x^{\ell+2}, \ldots , x^{2k}>} \\ &= 
\pI{<x^2 R,\ldots, x^{2\ell - |J|}R>} \cap  \pIJ{<x^{\ell+2}, \ldots, x^{2k}>}\\
& = 
\pIJ{<x^2 R,\ldots, x^{2\ell - |J|}R>} \cap  \pIJ{<x^{\ell+2}, \ldots, x^{2k}>}\\
 & =  \pIJ{<x^2R, \ldots , x^{2\ell - |J|}R> \cap  <x^{\ell+2}, \ldots, x^{2k}>}.
\end{align*}
The last equality is also a consequence of Lemma~\ref{lem:injection} since the
direct image of an intersection by an injective map is the intersection of
the direct images.

Since all the $x_i$'s are nonzero, the polynomials $x^{\ell+2}$ and $R$ are prime
to each other, this yields
\begin{align*}
<x^2R, \ldots , x^{2\ell - |J|}R> \cap  <x^{\ell+2}, \ldots, x^{2k}>
& = <x^{\ell +2}R, \ldots, x^{2\ell - |J|}R >.
\end{align*}
Therefore,
\begin{equation}
  \label{eq:intersection}
\pI{<x^2, \ldots, x^{2 \ell}>} \cap \pIJ{<x^{\ell+2}, \ldots , x^{2k}>} 
= \pIJ{< x^{\ell + 2 } R(x), \ldots, x^{2\ell - |J|}R(x)>}
\end{equation}
and this last space has dimension $\ell - |J|- 1$.
Finally, combining~(\ref{eq:dimCCIsquare}), (\ref{eq:dims})
and~(\ref{eq:intersection}), we get
$$
\dim \CC_I^2 = (2k - \ell -1)  + (2\ell -1) -(\ell - |J| -1) = 2k + |J| - 1.
$$

\newpage

\section{Proof of Lemma \ref{lem:structure}}
    \label{appendixBaldi}
    
Recall that $\Rm$ has rank $1$, then so does $\Rm \Pim^{-1}$ and there exist
$\av$ and $\bv$ in $\fq^n$ such that $ \Rm \Pim^{-1} = \bv^T \av$.
Set
$$\Pm \eqdef \Im + \Rm \Pim^{-1} = \Im + \bv^T \av.$$
 
\noindent We first need the following lemmas
\begin{lem}
  \label{lem:PQinv}
  The matrix $\Qm$ is invertible if and only if $\Pm$ is.
\end{lem}

\begin{proof}
We have $\Qm = \Pim + \Rm = (\Im + \Rm \Pim^{-1})\Pim = \Pm \Pim$,
which yields the proof.
\end{proof}

\begin{lem}
\label{lem:inverse}
The matrix $\Pm$ is invertible if and only if $\scp{\av}{\bv} \neq -1$.
In addition, if it is invertible, then
$$
\Pm^{-1} = \Im -\frac{1}{1+\scp{\av}{\bv}} \bv^T \av.
$$
\end{lem}

\begin{proof}
First, assume that $\scp{\av}{\bv}\neq -1$. Then,
\begin{eqnarray*}
\Pm \left( \Im -\frac{1}{1+\scp{\av}{\bv}} \bv^T \av \right) &= & \left( \Im + \bv^T \av \right) \left( \Im -\frac{1}{1+\scp{\av}{\bv}} \bv^T \av \right)\\
& = & \Im + \left( 1 - \frac{1}{1+\scp{\av}{\bv} }\right)\bv^T \av  -\frac{1}{1+\scp{\av}{\bv}} \bv^T \av \bv^T \av \\
& = &\Im + \frac{\scp{\av}{\bv}}{1+\scp{\av}{\bv}}\bv^T \av - \frac{\scp{\av}{\bv}}{1+\scp{\av}{\bv}}\bv^T \av\\
& = & \Im.
\end{eqnarray*}

To conclude the ``only if'' part of the proof, there remains to prove that $\Pm$ is non invertible for $\scp{\av}{\bv} = -1$.
Assume $\scp{\av}{\bv} = -1$, then
$$
\Pm^2 = \Im +2\bv^T \av + \bv^T \av \bv^T \av = \Im + (2 +\scp{\av}{\bv}) \bv^T \av = \Pm.
$$
Thus, in this situation, $\Pm$ is a projection distinct from $\Im$ and hence
is non invertible.
\end{proof}

\begin{proof}[Proof of Lemma~\ref{lem:structure}]
 Let $\cv$ be an element of $\Cpub$.
 Since
$$\Csec =  \Cpub \Qm =  \Cpub (\Pim + \Rm)= \Cpub (\Im + \Rm\Pim^{-1})\Pim.$$
We obtain
\begin{equation}
  \label{eq:Cpub=CsecPim1}
   \CC =\Csec \Pim^{-1}=  \Cpub \Pm \quad {\rm where}\quad  \Pm \eqdef \Im + \Rm \Pim^{-1}.
\end{equation}
Therefore
 $$
 \Cpub = (\Csec \Pim^{-1}) \Pm^{-1} = \CC \Pm^{-1}.
 $$
 From this, we obtain that there exists $\pv$ in $\CC$ such that
 $\cv = \pv \Pm^{-1}$. Thus, from Lemma~\ref{lem:inverse} 
we know that 
$\Pm^{-1} = \Im -\frac{1}{1+\scp{\av}{\bv}} \bv^T \av = \Im  + \lambdav^T \av$, which enables to write:
$$
\cv =  \pv \left(   \Im  + \lambdav^T \av \right)  = \pv+(\scp{\lambdav}{\pv})\av.
$$
\end{proof}

\begin{cor}
  \label{cor:inverse}
  Given $\uv, \vv \in \fq^n$ the map $\pv \mapsto \pv + (\scp{\uv}{\pv})\vv$
  is an automorphism of $\fq^n$ if and only if $\scp{\uv}{\vv}\neq -1$.
\end{cor}

 \newpage
\section{Proof of Proposition~\ref{prop:three}} \label{proof:dim3k_3}

This follows immediately from the fact that we can express
 $\zv_i$ in terms of the $\gv_j$'s, say
 $$
 \zv_i = \sum_{1 \leqslant  j \leqslant  k} a_{ij} \gv_j.
 $$
 We observe now that there exist three  relations between the 
$\zv_i \cwp \gv_j$'s:
\begin{eqnarray}
\sum_{1 \leqslant  j \leqslant  k} a_{2j} \zv_1 \cwp \gv_j - \sum_{1 \leqslant  j \leqslant  k} a_{1j} \zv_2 \cwp \gv_j 
= \zv_1 \cwp \zv_2 - \zv_2 \cwp \zv_1 &= & 0 \label{eq:crossproduct}\\
\sum_{1 \leqslant  j \leqslant  k} a_{3j} \zv_1 \cwp \gv_j - \sum_{1 \leqslant  j \leqslant  k} a_{1j} \zv_3 \cwp \gv_j 
= \zv_1 \cwp \zv_3 - \zv_3 \cwp \zv_1 &= & 0 \label{eq:crossproduct2}\\
\sum_{1 \leqslant  j \leqslant  k} a_{3j} \zv_2 \cwp \gv_j - \sum_{1 \leqslant  j \leqslant  k} a_{2j} \zv_3 \cwp \gv_j 
= \zv_2 \cwp \zv_3 - \zv_3 \cwp \zv_2 &= & 0 \label{eq:crossproduct3}
\end{eqnarray}
It remains to prove that the three obtained identities relating the $\zv_i \cwp \gv_j$'s
are independent under some conditions on the $\zv_i$'s. Actually, these relations are
independent if and only if the $\zv_i$'s
generate a space of dimension larger than or equal to $2$.
Indeed, sort the $\zv_1 \cwp \gv_j$'s as $\zv_1 \cwp \gv_1, \ldots ,
\zv_1 \cwp \gv_k, \zv_2 \cwp \gv_1, \ldots , \zv_2 \cwp \gv_k,
\zv_3 \cwp \gv_1, \ldots ,  \zv_3 \cwp \gv_k$.
Then the system defined by Equations~(\ref{eq:crossproduct}) to (\ref{eq:crossproduct3})
is defined by the $ 3 \times 3k$ matrix
$$
A:=\begin{pmatrix}
  a_{21} & \cdots & a_{2k} & -a_{11} & \cdots & -a_{1k} & 0 & \cdots & 0 \\
  a_{31} & \cdots & a_{3k} & 0 & \cdots & 0  & -a_{11} & \cdots & -a_{1k} \\
  0 & \cdots & 0  & -a_{31} & \cdots & -a_{3k} & a_{21} & \cdots & a_{2k}  
\end{pmatrix}.
$$
Then, $A$ has rank strictly less than $3$ if there exists a vector $\word{u} = (u_1, u_2, u_3)$
such that $\word{u}A = 0$ which is equivalent to the system
$$
\left\{
  \begin{array}{rcc}
    u_1 \zv_2 + u_2 \zv_3 & = & 0\\
    -u_1 \zv_1 - u_3 \zv_3 & = & 0\\
    - u_2 \zv_1 + u_3 \zv_2 & = & 0
  \end{array}
\right. 
$$
and such a system has a nonzero solution $\word{u}=(u_1, u_2, u_3)$
if and only if the $\zv_i$'s are pairwise collinear \textit{i.e.}
generate a subspace of dimension lower than or equal to $1$.
\end{document}